\documentclass[journal,twoside,web]{ieeecolor}
\usepackage{lcsys}
\usepackage{cite}
\usepackage{amsmath,amssymb,amsfonts}
\usepackage[ruled,vlined]{algorithm2e}
\usepackage{algpseudocode}
\usepackage{graphicx}
\usepackage{textcomp}
\usepackage{algpseudocode}
\newtheorem{definition}{\textbf{Definition}}
\newtheorem{proposition}{\textbf{Proposition}}
\newtheorem{remark}{\textbf{Remark}}

\newtheorem{lemma}{\textbf{Lemma}}
\newtheorem{assumption}{\textbf{Assumption}}
\def\BibTeX{{\rm B\kern-.05em{\sc i\kern-.025em b}\kern-.08em
		T\kern-.1667em\lower.7ex\hbox{E}\kern-.125emX}}
\begin{document}
	\title{\LARGE
		Generalized Stochastic Dynamic Aggregative Game for Demand-Side Management in Microgrids with Shared Battery}
	\author{Shahram Yadollahi$^{1}$, Hamed Kebriaei$^{1}$, \IEEEmembership{Senior Member, IEEE}, and Sadegh Soudjani$^{2}$, \IEEEmembership{Member, IEEE} 
		\thanks{$^{1}$ Shahram Yadollahi and Hamed Kebriaei are with the School of ECE, College of Engineering, University of Tehran, Tehran, Iran. Emails: (shahram.yadollahi@ut.ac.ir, kebriaei@ut.ac.ir)}%
		\thanks{$^{2}$ Sadegh Soudjani is with the School of Computing, Newcastle University, United Kingdom. Email:
			(Sadegh.Soudjani@newcastle.ac.uk)}}
	\maketitle
	\thispagestyle{empty} 
	\pagestyle{empty} 

\begin{abstract}
In this paper, we focus on modeling and analysis of demand-side management in a microgrid where agents utilize grid energy and a shared battery charged by renewable energy sources.
We model the problem as a generalized stochastic dynamic aggregative game with chance constraints that capture the effects of uncertainties in the renewable generation and agents’ demands.
Computing the solution of the game is a complex task due to probabilistic and coupling constraints among the agents through the state of charge of the shared battery.
We investigate the Nash equilibrium of this game under uncertainty considering both the uniqueness of the solution and the effect of uncertainty on the solution.
Simulation results demonstrate that the presented stochastic method is superior to deterministic methods.
\end{abstract}


\section{ Introduction}
\IEEEPARstart{M}{ulti}-agent coordination for energy systems has emerged as a highly effective approach to achieve energy savings and maintain stability in microgrid systems. To accomplish this objective, concepts and techniques from game theory has been utilized due to their ability to incorporate user behavior and optimization perspectives \cite{Game-Theoretic-Methods-microgrid-Basar}.


In this paper, we study demand-side management (DSM) in microgrids as a \textit{generalized stochastic dynamic aggregative game} with uncertainties in renewable energy and demand, using a shared battery and selfish residential consumers. Additionally, operational constraints were taken into account in modeling the problem.
To avoid suffering from peak loads, a term related to the cumulative exchange of power between the agents and the grid appears in the cost function of the agents as part of electricity tariff. This leads to the aggregative form of the proposed game model. Unlike previous studies, a more comprehensive form of stochastic constraints is considered in the form of chance constraints.
Due to the incorporation of a shared battery and accounting for the impact of uncertain sources, the state of charge (SoC) of the battery has stochastic dynamics shared between the agents. In addition, as we impose a constraint on the SoC of the shared battery, we have both dynamic and static stochastic coupling constraints among the agents, which shape the proposed game in a generalized, stochastic, and dynamic form. Then, we propose a series of reformulations and guaranteed under-approximations over the cost functions, stochastic dynamics, and chance constraints such that the game is converted into a static generalized aggregative form. Finally, we verify the conditions under which a Nash-seeking method can obtain the game's equilibrium point.

\noindent\textbf{Related Works.}
Game theory analyzes strategic interactions in demand-side management, modeling consumer behavior and optimizing incentives for efficient resource use.
The main types of games studied in this context are \textit{dynamic games} that focus on the dynamic relationships among the agents (e.g., shared battery resources, dynamic pricing, etc.), and \textit{generalized aggregative games}, which account for the coupling constraints among the agents.
Comprehensive reviews on the analysis and decomposition of dynamic games can be found in 
\cite{BasarZaccour2018,bacsar1998dynamic}.
Various theoretical approaches have been employed to solve dynamic games utilizing optimal control theory to study open-loop and closed-loop Nash equilibrium \cite{bacsar1998dynamic,jank2003existence,reddy2015open}.
For deterministic finite horizon discrete-time dynamic games, the state dynamic equation can be treated as a finite number of constraints which transforms the problem  into a generalized aggregative game \cite{abraham2019new}, 
however, increasing the number of constraints can impose high computational cost.
Stochastic dynamic games have also been studied for both continuous-time 
\cite{ huang2006large,zhao2022varepsilon}
and discrete-time systems
\cite{pachter2010discrete,velez2021open}. 
However, these works do not study chance constraints on the system's state or coupling constraint on the control inputs of agents. In contrast to previous studies that use almost-sure satisfaction of constraints \cite{Farzaneh2020,zou2019game}, we employ chance constraints due to their flexibility and potential for better outcomes with a predefined confidence. Imposing deterministic constraints on estimated random variables have also been considered in \cite{joshi-kebriaei-2,mohsenian2010autonomous,mediwaththe2015dynamic}, but this approach overlooks the inherent stochastic nature of the uncertainties, increasing the risk of practical implementation failures.
From an application perspective,
DSM is addressed through dynamic games with \cite{DynamicGame-transactionSmartGrid-2016,mediwaththe2019incentive}  or without  \cite{pilz2020dynamic,liu2016algorithmic} the presence of shared battery.
DSM has also been studied through generalized aggregative games
\cite{jo2020demand,estrella2019shrinking,joshi-kebriaei-2}.
In \cite{estrella2019shrinking,joshi-kebriaei-2},
 the impact of shared dynamics was not taken into account.
 These works studied the game in deterministic setups without considering uncertainty sources in microgrid systems.
In this paper, we are dealing with a DSM game with stochastic shared dynamics of battery and coupling chance constraint on the state, referred to as a generalized stochastic dynamic aggregative game.

\smallskip
\noindent\textbf{Contributions.}
The contributions of this paper are:
\begin{itemize}
    \item Introducing a novel framework for demand-side management as a generalized stochastic dynamic aggregative game, featuring constraints in the form of stochastic dynamics and chance constraints.
    \item Providing a guaranteed under-approximation for the game in the form of a
     generalized static aggregative game with deterministic constraints, in which, the feasible set of the resulted game is a subset of the feasible set of the original game.
    \item Analyzing the existence and uniqueness of the generalized Nash equilibrium (GNE) for the game, and proposing a semi-decentralized algorithm for Nash seeking.
\end{itemize}

\smallskip

\section{review on some mathematical concepts and tools}\label{sec.preliminaries}
\subsection{Chernoff-Hoeffding inequality}
{\normalsize
Consider a sequence of random variables, $Z_i$, where $i$ ranges from $1$ to $n$, each with bounded support between $a_i$ and $b_i$. Let each of these variables possess an expected value, $\mathbf{E}\{Z_i\}$,
belonging to the moment interval, $\mathbb{M}_i$. Now, let's introduce a new random variable, $Z$, defined as the sum of all $Z_i$ from $i=1$ to $n$. 
By employing the \textit{Chernoff-Hoeffding inequality}, we have the following inequalities:

{\small
\begin{align*}
&\mathbf{Pr}\left\{Z - \mathbf{E}\left\{Z\right\} \leq -\zeta \right\} \leq \mathrm{exp}\left( \frac{-\zeta^2}{\nu \sum_{i = 1}^{n}\left(b_{i} - a_{i}\right)^2}\right)\, \ \forall \zeta \geq 0\\
&\mathbf{Pr}\left\{Z - \mathbf{E}\left\{Z\right\} \geq \zeta \right\} \leq \mathrm{exp}\left( \frac{-\zeta^2}{\nu \sum_{i = 1}^{n}\left(b_{i} - a_{i}\right)^2}\right)\ \ \forall \zeta \geq 0 \ .
\end{align*}
}
In this context, $\nu = \chi(\hat{G})/2$ is a positive constant as referred in \cite{janson2004large}. Here, $\hat{G}$ signifies the undirected dependency graph of the random variables $Z_1, ..., Z_n$, and $\chi(\hat{G})$ represents the chromatic number of this graph. The chromatic number of $\hat{G}$, is defined as the minimum number of colors required to color the vertices of $\hat{G}$ such that no adjacent vertices share the same color.
For instances when the variables are independent, the chromatic number of the graph $\hat{G}$ is one.    

\subsection{Operator theoretic definitions} 
\begin{itemize}
	\item For a closed set $S \subseteq \mathbb{R}^n$, the mapping $\mathrm{proj}_{S}:\mathbb{R}^{n} \longrightarrow S$ denotes the \textit{projection} onto $S$, i.e., $\mathrm{proj}_{S}(x) = \mathrm{argmin}_{y \in S} \|y-x\|$.
	\item A set-valued mapping $\mathcal{F}: \mathbb{R}^n \rightrightarrows \mathbb{R}^n$ is \textit{$l-$Lipschiz continuous}, with $l > 0$, if $\|u-v\| \leq l\|x-y\|$ for all $x,y \in \mathbb{R}^{n}$,$u \in \mathcal{F}(x), v \in \mathcal{F}(y)$.
	\item A set-valued mapping $\mathcal{F}: \mathbb{R}^n \rightrightarrows \mathbb{R}^n$ is \textit{monotone} if $(u-v)^{T}(x-y) \geq 0$ for all $x \neq y \in \mathbb{R}^{n}, u \in \mathcal{F}(x), v \in \mathcal{F}(y)$.
	\item A set-valued mapping $\mathcal{F}: \mathbb{R}^n \rightrightarrows \mathbb{R}^n$ is \textit{strictly monotone} if $(u-v)^{T}(x-y) > 0$ for all $x \neq y \in \mathbb{R}^{n}, u \in \mathcal{F}(x), v \in \mathcal{F}(y)$.
	\item A set-valued mapping $\mathcal{F}: \mathbb{R}^n \rightrightarrows \mathbb{R}^n$ is $\eta-$\textit{strongly monotone} with $\eta > 0$, if $(u-v)^{T}(x-y) \geq \eta \|x-y\|^2$ for all $x, y \in \mathbb{R}^{n}, u \in \mathcal{F}(x), v \in \mathcal{F}(y)$.\\
\end{itemize}
\subsection{Generalized Variational Inequality}
Consider a closed convex set $S \subset \mathbb{R}^{n}$, a set-valued mapping
$\Psi : S \rightrightarrows \mathbb{R}^{n}$, and a single-valued mapping $\psi : S \rightarrow \mathbb{R}^{n}$. The generalized variational inequality problem $\mathrm{GVI}(S,\Psi)$, seeks to find $\mathbf{x}^{*} \in S $ and $\mathbf{g}^{*} \in \Psi(\mathbf{x}^{*})  $ the following condition holds for all $\mathbf{x} \in S$:
\begin{equation*}
(\mathbf{x} - \mathbf{x}^{*})^{T}\mathbf{g}^{*} \geq 0 . 
\end{equation*}
If $\Psi(\mathbf{x}) = {\psi(\mathbf{x})}$ for all $\mathbf{x} \in S$, then the generalized variational inequality problem $\mathrm{GVI}(S,\Psi)$ reduces to the standard variational inequality problem $\mathrm{VI}(S,\psi)$.\cite{belgioioso2018projected}

\subsection{Gershgorin circle theorem}
The Gershgorin circle theorem may be used to bound the spectrum of a square matrix.\\
Let $A$ be a complex $n \times n$, with entries $a_{ij}$. For $i = 1,\dots,n$ let $R_{i}$ be the sum of absolute values of non-diagonal entries in $i-$th row:
\begin{align*}
    R_{i} = \sum_{j = 1,j \neq i}^{n}|a_{ij}|
\end{align*}
Let $D(a_{ii},R_{i}) \subset \mathbb{C}$ , be a closed disc centered at $a_{ii}$ with radius $R_{i}$. Such a disc is called Gershgorin disc.
\begin{lemma}
Every eigenvalue of $A$ lies within at least one of the Gershgorin discs $D(a_{ii},R_{i})$.
\end{lemma}
}

\noindent\textbf{Notations.}
$\mathbb{R}$ denotes the set of real numbers.
$\mathbf{0}_{\tau}(\mathbf{1}_{\tau})$ denotes a vector with dimention $\tau \times 1$ that all elements equal to $0 (1)$. $\mathbf{1}_{\tau \times \tau}$ denotes a matrix with dimention $\tau \times \tau$ that all elements equal to $1$. 
$\mathbf{I}_{\tau}$ denotes a $\tau \times \tau$ identity matrix. $A \otimes B$ denotes the Kronecker product between matrices $A$ and $B$. Suppose that we have $\mathit{N}$ vectors, $\mathbf{x}_{1},...,\mathbf{x}_{N} \in \mathbb{R}^{n}$, then we define $\mathbf{x} \triangleq \left[\mathbf{x}_{1}^{T},...,\mathbf{x}_{N}^{T}\right]^{T}$. 
$\mathbf{M}_{\tau}$ is a $\tau \times \tau$ lower triangular matrix  such that $M_{\tau}(i,j) = 1$ if $i\leq j$, and zero otherwise.

\section{System Model}\label{sec.system-model}
We investigate a grid-connected community microgrid as shown in Figure~\ref{fig.system model}, which comprises $N$ selfish residential households serving as agents in the system.
Each household has the capacity to meet its own energy requirements via the power grid and a shared battery, which is recharged using renewable resources. Furthermore, we incorporate in our model the uncertainty in both the renewable energy sources and the energy demand.
The microgrid operates under a tariff scheme with a retailer communicating the tariff information to the households. Each household uses this information to decide its battery discharging profile for minimizing the total cost of electricity consumption, which includes the cost of purchasing electricity from the retailer and the cost of using the battery to store and discharge energy.
We consider the system where agents interact solely with a coordinator. All the decisions are taken over the day-ahead horizon. 

\subsection{Shared Battery Model}
Battery plays an important role in microgrid systems by providing energy storage solutions that can balance energy supply and demand,  
and contribute in peak load shaving by proving a grid-free source of energy in peak hours.
We consider that the battery is charged only through renewable energy and discharged by consumption of the agents from the battery. The state of charge (SoC) of the shared battery $x^{t}$
is considered to have the following dynamics
\begin{equation}\label{SoC dynamics of the battery}
x^{t+1} = x^{t} + \eta \Delta t \left[r^{t} - \sum_{j= 1}^{N}u_{j}^{t} \right].
\end{equation}
\begin{figure}[ht]
\centering
\includegraphics[height = 0.6\columnwidth,width = \columnwidth]{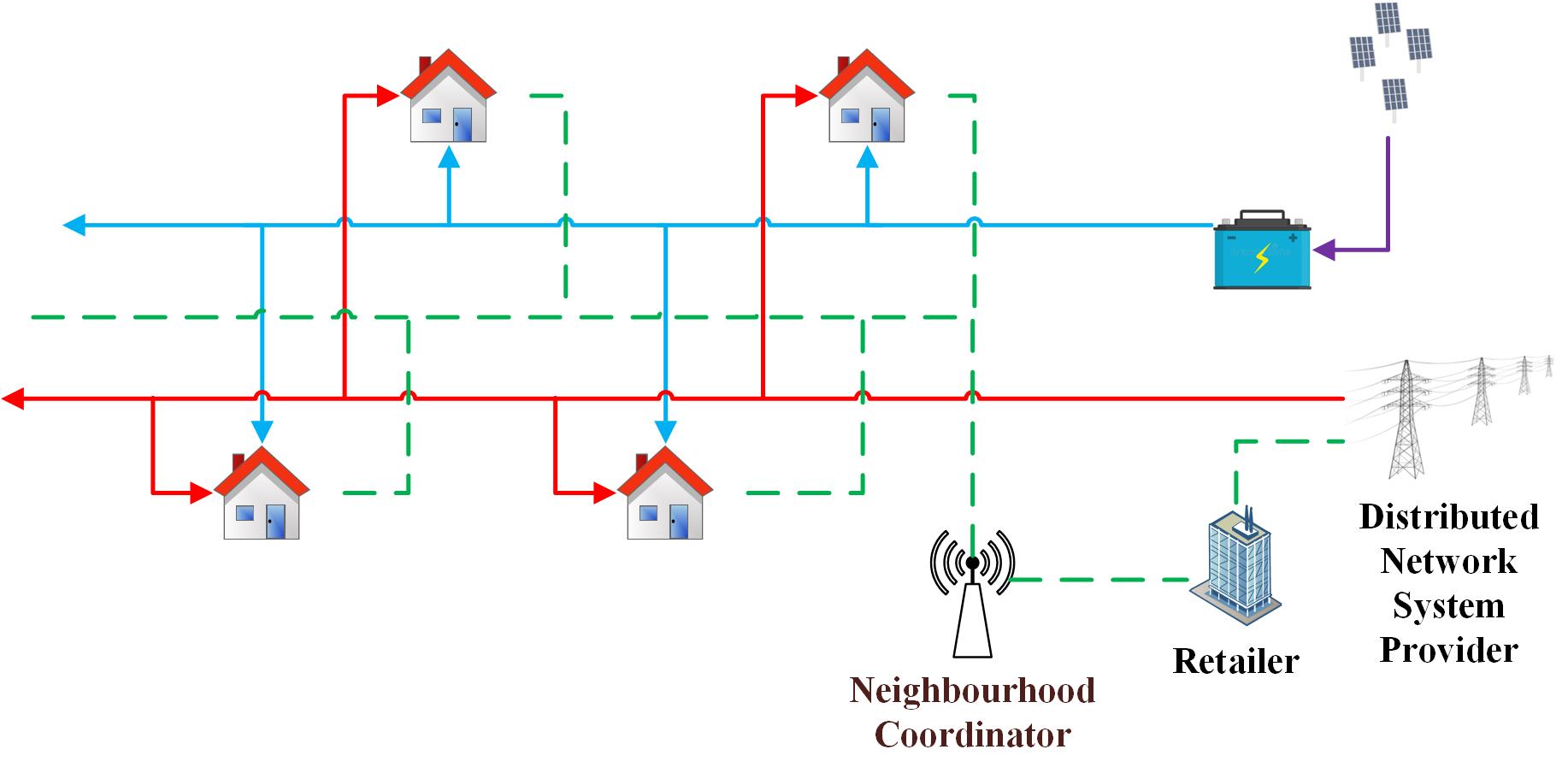}
\caption{System model. Red, blue and purple solid lines show respectively the energy supplied by the \textit{Distributed Network Provider}, \textit{shared battery} and \textit{renewable energy}. Green dashed line shows \textit{Essential Information Exchange (e.g. demands, common tariff)}.}
\label{fig.system model}
\vspace{-0.5cm}
\end{figure}
At time $t$, $u_{j}^{t}$ is the discharging decision of the battery by the $j^{\text{th}}$ agent and $r^{t}$ is the power generated by renewable energy sources, which is unknown and stochastic.
The parameter $\eta$ is the charging/discharging efficiency of the shared battery, and $\Delta t$ is the sampling time.

Due to the stochastic input in \eqref{SoC dynamics of the battery}, we consider the following chance constraints on the SoC and its final value:
\begin{align}
\label{Soc constraint}
	& \mathbf{Pr} \{  \underline{x} \leq x^{t} \leq \bar{x}\} \geq  1- \delta _{x}^{t} \hspace{0.5cm} \forall t \in \left \{0,\dots,\tau  \right \},\\
 & \mathbf{Pr}\{|  x^{\tau}-r | \leq \epsilon \} \geq  1-\delta_{x}^{\tau,\mathrm{final}},
 \label{reformulation last time SoC}
\end{align}
where $\underline{x}, \bar{x}, \delta _{x}^{t}$, and $\delta_{x}^{\tau,\mathrm{final}}$ are positive constants in $[0,1]$ such that $\underline{x} < \bar{x}$,
$r\in \left[\underline{x},\bar{x}\right]$ is a positive constant, and $\epsilon\in\left(0,\textrm{min}\left\{r-\underline{x},\bar{x}-r\right\}\right)$. 
The chance constraint in \eqref{reformulation last time SoC} guarantees that the SoC at the final time step lies within a specific range around $r$ with a certain confidence.
We also have the following constraint:
\begin{equation}
0 \leq u_{i}^{t} \leq \bar{u}_{i} \hspace{0.5cm} \forall t \in \left \{0,\dots,\tau-1  \right \}. \label{u constaint1}
\end{equation}
%


\subsection{Power exchange model}

The load balance equation for household $i$ at time $t \in  \{0,...,\tau-1\} $ can be written as $g_{i}^{t} =  d_{i}^{t} - u_{i}^{t}$,
%
where $g_{i}^{t}$ is power exchange of the $i^{\text{th}}$ agent with the grid and $d_{i}^{t}$ is its stochastic power demand.
Let $g^{t} \triangleq \sum_{i=1}^{N} g_{i}^{t}$. The retailer imposes the following constraints on the allowable power exchange with the community of the microgrid:
\begin{equation}\label{retailer constraint}
	\mathbf{Pr}\left\{0 \leq g^{t} \leq \bar{g} \right\} \geq 1- \delta _{g}^{t}\quad \forall t \in \{0,\dots,\tau-1\},
\end{equation}
where $\bar{g}$ is the maximum power supply of the retailer, and $\delta_{g}^{t}\in\left[0,1\right]$.
\subsection{Electricity tarrifs}\label{sec.SystemModel-tarrif}
All the agents in the neighbourhood are billed using common electricity tariffs, modelled as
\begin{equation}
\label{tarrif-calculation}
    \pi(g^{t}) = K_{ToU}^{t} + \sum_{j=1}^{N} g_{j}^{t} k_{c}^{N}\;,
\end{equation}
where $\pi(g^{t})$ is the common electricity tariff for agents at time $t$ and $K_{ToU}^{t}$ is the conventional time-of-use pricing tariff that could depend on hour of the day. The positive constant $k_{c}^{N}$ influences the
cost in proportion to the peak power consumption creates a balance between the microgrid's objectives and ensuring fairness, peak shaving, and stability in the tariff system. As mentioned in \cite{Farzaneh2020,chavali2014distributed}, \( k_{c}^{N} \) is selected to be inversely proportional to the number of agents (this dependency is denoted by the superscript $N$), which performs normalization in aggregative term of grid electricity use of agents. The tariff
function \eqref{tarrif-calculation} is based on the hypothesis that increasing the price
of electricity at times of peak demand will motivate a rational
household to schedule the shared battery such that the community
peak can be shaved \cite{Shahidehpour-tariff,Carlson-tariff}.
\subsection{Cost Function of Each Agent in Model}
The cost function of each agent (household) is
\begin{align}\label{cost-pure}
	J_{i} = \mathbf{E} \Bigg\{\sum_{t=0}^{\tau-1}    \Bigg[& \pi\left(g^t\right) g_{i}^{t}+
	 \sum_{j=1}^{N} \left(\alpha^{dch} \left(u_{j}^{t}\right)^{2} + \beta^{dch} u_{j}^{t}\right) \Bigg ] \Bigg\},
\end{align}
where $\alpha^{dch}, \beta^{dch}$ are positive constants.
The cost function has a conventional form used frequently in the literature (see e.g., \cite{joshi-kebriaei-2,Farzaneh2020}). It consists of two parts: one related to the cost of electricity and the other to battery degradation, which serves as a proxy for the shared battery lifespan.

Let $\mathbf{u}_{i}=[ u_{i}^{0}, u_{i}^{1},\dots,u_{i}^{\tau-1}]^{T}$. We are now equipped to delineate the game-theoretic setup for our demand-side management model, as
{\small
     \begin{align*}
      \mathcal{G} = \begin{cases}
        \textbf{Players:}\  \textrm{A set of residential agents}\  \mathcal{N} = \{1,2,\dots,N\}\\
        \textbf{Strategies of Agents:}\  \mathbf{u}_{i} \hspace{1cm} i \in \mathcal{N}\\
        \textbf{Cost Functions:}\               
                        J_{i}\left(\mathbf{u}_{i},\mathbf{u}_{-i}\right) 
        \hspace{0.75cm} i \in \mathcal{N}\\
        \textbf{Stochastic Dynamic:}\  \eqref{SoC dynamics of the battery}\\
        \textbf{Constraints:}
        \begin{cases}
        \textbf{Local:} \ \eqref{u constaint1}\\
        \textbf{Coupling}\ \textbf{(Chance Constraints):}\ \eqref{Soc constraint},\eqref{reformulation last time SoC}, 
        \eqref{retailer constraint}
        \end{cases}
        \end{cases}           
    \end{align*}
}
As demonstrated above,
$\mathcal{G}$ evidently corresponds to a generalized stochastic dynamic aggregative game.

\section{Approximation with Generalized Static Aggregative
Game}
\label{sec.model-reformulation}
We raise the following assumption for constructing an under-approximation of the feasible set of the game.
\begin{assumption}
The initial state $x^{0}$ is known and $\underline{x} \leq x^{0}  \leq \bar{x}$. Additionally, $r^t$ and $d_{i}^{t}$ are random variables with bounded support $[a^t,b^t]$ and $[c_{i}^{t},f_{i}^{t}]$, and their mean values are available to each agent. Random variables $r^t$ and $d_{i}^{t}$ could in general be dependent with known dependency graph. 
\end{assumption}

The above assumption is based on the fact that the power generated by renewable energy sources and  power demand of agents are bounded.
Let us denote the input constraints in
\eqref{u constaint1} in vector format as $\mathbf{0}_{\tau}   \leq \mathbf{u}_{i}  \leq 
\bar{u}_{i}\mathbf{1}_{\tau}$.

\begin{proposition}
\label{pro. Soc reformulation}
The feasible domain of the chance constraint \eqref{Soc constraint} can be under-approximated with the deterministic inequalities
\begin{align}
\label{Soc constraint Concatenate leq}
& \left(\underline{x}-x^{0}\right)\mathbf{1}_{\tau} - \rho \mathbf{M}_{\tau} \mu_{r} + \mathbf{q}_{x,1}+\sum_{j=1}^{N}\rho\mathbf{M}_{\tau}\mathbf{u}_{j} \leq 0,\\
& \left(x^{0} - \bar{x}\right)\mathbf{1}_{\tau} + \rho \mathbf{M}_{\tau} \mu_{r} + \mathbf{q}_{x,2}
-\sum_{j=1}^{N}\rho\mathbf{M}_{\tau}\mathbf{u}_{j} \leq 0, \label{Soc constraint Concatenate geq}
\end{align}
where
$\mu_{r} = \left[
\mathbf{E}\left\{r^{0}\right\},\mathbf{E}\left\{r^{1}\right\},\dots,\mathbf{E}\left\{r^{\tau - 1}\right\}\right]^T$,
$\mathbf{q}_{x,1} = \left[q_{x,1}^1,q_{x,1}^2,\dots,q_{x,1}^\tau \right]^T$,
$\mathbf{q}_{x,2} = \left[q_{x,2}^1,q_{x,2}^2,\dots,q_{x,2}^\tau \right]^T$,
with $q_{x,1}^t \triangleq \sqrt{-\rho^2 \nu_{r}^{t}\sum_{k=0}^{t-1}\left(b^{k} - a^{k}\right)^{2} \mathrm{Ln}(\tilde{\delta}_{x}^{t})}$,
$q_{x,2}^t \triangleq \sqrt{-\rho^2 \nu_{r}^{t}\sum_{k=0}^{t-1}\left(b^{k} - a^{k}\right)^{2} \mathrm{Ln}(\tilde{\delta}_{x}^{t}- \delta_{x}^{t})}$
, $0 \leq \tilde{\delta}_{x}^{t} \leq 1$,
$0 \leq \delta_{x}^{t} - \tilde{\delta}_{x}^{t}  \leq 1$, and
$\rho \triangleq \eta \Delta t$. Moreover, $\nu_{r}^{t} = \chi(\hat{G}_{r}^{t})/2$ with $\hat{G}_{r}^{t}$ being the undirected dependency graph of the random variables $r^0,\ldots, r^{t-1}$, and $\chi(\hat{G}_{r}^{t})$ represents the chromatic number of this graph (see \cite{janson2004large}).
\end{proposition}

\begin{proof}
The proof is based on deriving from 
\eqref{SoC dynamics of the battery} the explicit form of $x^t$ as
\begin{equation}\label{general state}
	x^{t} = x^{0} + \sum_{k=0}^{t-1}\left \{ -\rho \left ( \sum_{j=1}^{N}  u_{j}^{k}     \right)   + \rho r^{k} \right \}
\end{equation}
and using \textit{Chernoff-Hoeffding inequality} \cite{janson2004large}.\\
It is obvious that:
\begin{align*}
    &\mathbf{Pr}\{\underline{x} \leq x^t \leq \bar{x} \} \geq 1-\delta_{x}^{t}
    \Longleftrightarrow \\ &\mathbf{Pr}\left\{\left(\underline{x} \leq x^t\right) \bigcap \left(x^t \leq \bar{x} \right) \right\} \geq 1-\delta_{x}^{t} \Longleftrightarrow \\
    &\mathbf{Pr}\left\{\left(x^t \leq \underline{x}\right) \bigcup \left( \bar{x}\leq x^t \right) \right\} \leq \delta_{x}^{t} \Longleftarrow\\&
    \begin{cases}
    	\mathbf{Pr}\left\{x^t \leq \underline{x}  \right\} \leq \tilde{\delta}_{x}^{t} \\
    	\mathbf{Pr}\left\{ \bar{x}\leq x^t  \right\} \leq \delta_{x}^{t} - \tilde{\delta}_{x}^{t}
    \end{cases}.\\
\end{align*}
where $\tilde{\delta}_{x}^{t}$ is real value in $[0,1]$ such that $0 \leq \delta_{x}^{t} - \tilde{\delta}_{x}^{t} \leq 1$. Consequently, we establish certain sufficient conditions for $\mathbf{Pr}\{\underline{x} \leq x^t \leq \bar{x} \} \geq 1-\delta_{x}^{t}$. Initially, we commence with the condition $\mathbf{Pr}\left\{x^t \leq \underline{x}  \right\} \leq \tilde{\delta}_{x}^{t}$.
{\footnotesize
\begin{align*}
	&\mathbf{Pr}\left\{ x^0 + \sum_{k=0}^{t-1}\sum_{j=1}^{N}(-\rho )u_{j}^{k} + \sum_{k=0}^{t-1}\rho r^k  \leq \underline{x}\right\} \leq \tilde{\delta}_{x}^{t}\Longleftrightarrow\\
	&\mathbf{Pr}\left\{ x^0 + \sum_{k=0}^{t-1}\sum_{j=1}^{N}(-\rho)u_{j}^{k} + \sum_{k=0}^{t-1}\rho r^k -\sum_{k=0}^{t-1}\rho\mu_{r}^{k} \leq \underline{x} -\sum_{k=0}^{t-1}\rho\mu_{r}^{k} \right\} \leq \tilde{\delta}_{x}^{t}\\&\Longleftrightarrow\\
	&\mathbf{Pr}\left\{ \sum_{k=0}^{t-1}\rho r^k -\sum_{k=0}^{t-1}\rho \mu_{r}^{k} \leq \underline{x} - x^0  - \sum_{k=0}^{t-1}\sum_{j=1}^{N}(-\rho)u_{j}^{k} -\sum_{k=0}^{t-1}\rho\mu_{r}^{k} \right\} \leq \tilde{\delta}_{x}^{t}\\&\Longleftrightarrow\\
	&\mathbf{Pr}\left\{ \sum_{k=0}^{t-1}r^k -\sum_{k=0}^{t-1}\mu_{r}^{k} \leq \frac{1}{\rho}\left(\underline{x} - x^0  - \sum_{k=0}^{t-1}\sum_{j=1}^{N}(-\rho)u_{j}^{k} -\sum_{k=0}^{t-1}\rho\mu_{r}^{k} \right)\right\} \\&\leq \tilde{\delta}_{x}^{t}
\end{align*}
}
In order to employ the \textit{Chernoff-Hoeffding inequality}, it is essential that the following condition be satisfied:
\begin{align*}
    \left(\underline{x} - x^0  - \sum_{k=0}^{t-1}\sum_{j=1}^{N}(-\rho)u_{j}^{k} -\sum_{k=0}^{t-1}\rho\mu_{r}^{k} \right) \leq 0
\end{align*}
If the aforementioned condition is met, we can, based on the Chernoff-Hoeffding inequality, state the following:
{\footnotesize
\begin{align*}
&\mathbf{Pr}\left\{ \sum_{k=0}^{t-1}r^k -\sum_{k=0}^{t-1}\mu_{r}^{k} \leq -\frac{1}{\rho}\left(-\underline{x} + x^0  + \sum_{k=0}^{t-1}\sum_{j=1}^{N}(-\rho)u_{j}^{k} +\sum_{k=0}^{t-1}\rho\mu_{r}^{k} \right)\right\} \\
&\leq \mathrm{exp}\left( 
\frac{-\left[\frac{1}{\rho}\left(-\underline{x} + x^0  + \sum_{k=0}^{t-1}\sum_{j=1}^{N}(-\rho)u_{j}^{k} +\sum_{k=0}^{t-1}\rho\mu_{r}^{k} \right)\right]^2}{\nu^t \sum_{k =0}^{t-1}(b^k - a^k)^2} \right)
\end{align*}}
So it is necessary to have:
{\footnotesize
\begin{align*}
&\mathrm{exp}\left( 
\frac{-\left[\frac{1}{\rho}\left(-\underline{x} + x^0  + \sum_{k=0}^{t-1}\sum_{j=1}^{N}(-\rho)u_{j}^{k} +\sum_{k=0}^{t-1}\rho\mu_{r}^{k} \right)\right]^2}{\nu^t \sum_{k =0}^{t-1}(b^k - a^k)^2} \right)  \leq \tilde{\delta}_{x}^{t}\\& \Longleftrightarrow\\
&\left(-\underline{x} + x^0  + \sum_{k=0}^{t-1}\sum_{j=1}^{N}(-\rho)u_{j}^{k} +\sum_{k=0}^{t-1}\rho\mu_{r}^{k} \right)^2 \geq -\rho^2 \nu^t \mathrm{Ln}\tilde{\delta}_{x}^{t}\sum_{k =0}^{t-1}(b^k - a^k)^2
\end{align*}}
Clearly, the term $-\rho^2 \nu^t \mathrm{Ln}\tilde{\delta}_{x}^{t}\sum_{k =0}^{t-1}(b^k - a^k)^2$ is non-negative, thus we can state:
{\small
\begin{align*}
    &-\underline{x} + x^0  + \sum_{k=0}^{t-1}\sum_{j=1}^{N}(-\rho)u_{j}^{k} +\sum_{k=0}^{t-1}\rho\mu_{r}^{k} \\&\geq \sqrt{-\rho^2 \nu^t \mathrm{Ln}\tilde{\delta}_{x}^{t}\sum_{k =0}^{t-1}(b^k - a^k)^2}\\
    &\mathrm{Or}\\
    &-\underline{x} + x^0  + \sum_{k=0}^{t-1}\sum_{j=1}^{N}(-\rho)u_{j}^{k} +\sum_{k=0}^{t-1}\rho\mu_{r}^{k} \\& \leq - \sqrt{-\rho^2 \nu^t \mathrm{Ln}\tilde{\delta}_{x}^{t}\sum_{k =0}^{t-1}(b^k - a^k)^2}    
\end{align*}}
But $-\underline{x} + x^0  + \sum_{k=0}^{t-1}\sum_{j=1}^{N}(-\rho)u_{j}^{k} +\sum_{k=0}^{t-1}\rho\mu_{r}^{k} \geq \sqrt{-\rho^2 \nu^t \mathrm{Ln}\tilde{\delta}_{x}^{t}\sum_{k =0}^{t-1}(b^k - a^k)^2}$ is only valid, base on $\left(\underline{x} - x^0  - \sum_{k=0}^{t-1}\sum_{j=1}^{N}(-\rho)u_{j}^{k} -\sum_{k=0}^{t-1}\rho\mu_{r}^{k} \right) \leq 0$. Hence, we arrive at the following conclusion:
{\small
\begin{align*}
	&\underline{x} - x^0  - \sum_{k=0}^{t-1}\sum_{j=1}^{N}(-\rho)u_{j}^{k} -\sum_{k=0}^{t-1}\rho\mu_{r}^{k} \\&+ \sqrt{-\rho^2 \nu^t \mathrm{Ln}\tilde{\delta}_{x}^{t}\sum_{k =0}^{t-1}(b^k - a^k)^2}\leq 0
\end{align*}}
Now, we commence with the condition $\mathbf{Pr}\left\{  \bar{x} \leq x^t \right\}  \leq  \delta_{x}^{t} -~ \tilde{\delta}_{x}^{t}$.
{\footnotesize
\begin{align*}
  &\mathbf{Pr}\left\{  \bar{x} \leq x^t \right\}  \leq \delta_{x}^{t} - \tilde{\delta}_{x}^{t} \Longleftrightarrow\\&
  \mathbf{Pr}\left\{ \bar{x} \leq x^0 + \sum_{k=0}^{t-1}\sum_{j=1}^{N}(-\rho)u_{j}^{k} + \sum_{k=0}^{t-1}\rho r^k \right\}  \leq \delta_{x}^{t} - \tilde{\delta}_{x}^{t} 
  \Longleftrightarrow \\& \mathbf{Pr}\left\{ \bar{x} - x^0 - \sum_{k=0}^{t-1}\sum_{j=1}^{N}(-\rho)u_{j}^{k} - \sum_{k=0}^{t-1}\rho\mu_r^k \leq \sum_{k=0}^{t-1}\rho r^k - \sum_{k=0}^{t-1}\rho\mu_r^k \right\}  \\&\leq \delta_{x}^{t} - \tilde{\delta}_{x}^{t}
  \Longleftrightarrow \\&\mathbf{Pr}\left\{  \sum_{k=0}^{t-1}r^k - \sum_{k=0}^{t-1}\mu_r^k \geq \frac{1}{\rho}\left( \bar{x} - x^0 - \sum_{k=0}^{t-1}\sum_{j=1}^{N}(-\rho)u_{j}^{k} - \sum_{k=0}^{t-1}\rho\mu_r^k \right) \right\}  \\&\leq \delta_{x}^{t} - \tilde{\delta}_{x}^{t} 
\end{align*}}
In order to employ the \textit{Chernoff-Hoeffding inequality}, it is essential that the following condition be satisfied:
\begin{align*}
   \bar{x} - x^0 - \sum_{k=0}^{t-1}\sum_{j=1}^{N}(-\rho)u_{j}^{k} - \sum_{k=0}^{t-1}\rho\mu_r^k \geq 0 
\end{align*}
If the aforementioned condition is met, we can, based on the Chernoff-Hoeffding inequality, state the following:
{\footnotesize
\begin{align*}
    &\mathbf{Pr}\left\{  \sum_{k=0}^{t-1}r^k - \sum_{k=0}^{t-1}\mu_r^k \geq \frac{1}{\rho}\left( \bar{x} - x^0 - \sum_{k=0}^{t-1}\sum_{j=1}^{N}(-\rho)u_{j}^{k} - \sum_{k=0}^{t-1}\rho\mu_r^k \right) \right\} \\
    &\leq \mathrm{exp} \left( \frac{-\left[  \frac{1}{\rho}\left( \bar{x} - x^0 - \sum_{k=0}^{t-1}\sum_{j=1}^{N}(-\rho)u_{j}^{k} - \sum_{k=0}^{t-1}\rho\mu_r^k \right) \right]^2}{\nu^t \sum_{k=0}^{t-1}\left( b^k - a^k\right)^2 } \right)
\end{align*}}
Like the previous part, it is necessary to have:
\begin{align*}
    &\left( \bar{x} - x^0 - \sum_{k=0}^{t-1}\sum_{j=1}^{N}(-\rho)u_{j}^{k} - \sum_{k=0}^{t-1}\rho\mu_r^k \right)^2 \\ &\geq -\rho^2\nu^t \mathrm{Ln} \left(  \delta_{x}^{t} - \tilde{\delta}_{x}^{t} \right)\sum_{k=0}^{t-1}\left( b^k - a^k\right)^2    
\end{align*}
It is obvious that $-\rho^2\nu^t \mathrm{Ln} \left( 1 - \delta_{t} - \tilde{\delta}_{t} \right)\sum_{k=0}^{t-1}\left( b^k - a^k\right)^2$ is non-negative, so:
\begin{align*}
 &\bar{x} - x^0 - \sum_{k=0}^{t-1}\sum_{j=1}^{N}(-\rho)u_{j}^{k} - \sum_{k=0}^{t-1}\rho\mu_r^k \\&\geq \sqrt{-\rho^2\nu^t \mathrm{Ln} \left(  \delta_{x}^{t} - \tilde{\delta}_{x}^{t} \right)\sum_{k=0}^{t-1}\left( b^k - a^k\right)^2}\\
 & \mathrm{Or}\\
 &\bar{x} - x^0 - \sum_{k=0}^{t-1}\sum_{j=1}^{N}(-\rho)u_{j}^{k} - \sum_{k=0}^{t-1}\rho\mu_r^k \\&\leq - \sqrt{-\rho^2\nu^t \mathrm{Ln} \left( \delta_{x}^{t} - \tilde{\delta}_{x}^{t} \right)\sum_{k=0}^{t-1}\left( b^k - a^k\right)^2}
\end{align*}
But $\bar{x} - x^0 - \sum_{k=0}^{t-1}\sum_{j=1}^{N}(-\rho)u_{j}^{k} - \sum_{k=0}^{t-1}\rho\mu_r^k \geq \sqrt{-\rho^2\nu^t \mathrm{Ln} \left( \delta_{x}^{t} - \tilde{\delta}_{x}^{t} \right)\sum_{k=0}^{t-1}\left( b^k - a^k\right)^2}$ is just valid, base on $\bar{x} - x^0 - \sum_{k=0}^{t-1}\sum_{j=1}^{N}(-\rho)u_{j}^{k} - \sum_{k=0}^{t-1}\rho\mu_r^k \geq 0 $. Hence, we arrive at the following conclusion:
{\small
\begin{align*}
    &-\bar{x} + x^0 + \sum_{k=0}^{t-1}\sum_{j=1}^{N}(-\rho)u_{j}^{k} \\&+ \sum_{k=0}^{t-1}\rho\mu_r^k+\sqrt{-\rho^2\nu^t \mathrm{Ln} \left( \delta_{x}^{t} - \tilde{\delta}_{x}^{t} \right)\sum_{k=0}^{t-1}\left( b^k - a^k\right)^2} \leq 0
\end{align*}}
\end{proof}







\begin{proposition}
The feasible domain of the chance constraint \eqref{reformulation last time SoC} can be under-approximated with the deterministic
inequalities
\begin{align}
&r\!-\!\epsilon\!-\!x^{0}\! - \!\rho \mathbf{1}_{\tau}^{T} \mu_{r} \!+\!q_{x,\mathrm{final},1}\!+\!\sum_{j=1}^{N}\rho\mathbf{1}_{\tau}^{T}\mathbf{u}_{j} \leq 0,  \label{deterministic form of reformulation last time SoC leq}\\
&x^{0}\! -\! r \!-\! \epsilon\! +\! \rho \mathbf{1}_{\tau}^{T} \mu_{r} \!+\!q_{x,\mathrm{final},2} \!+\!\sum_{j=1}^{N}\left(-\rho\mathbf{1}_{\tau}^{T}\right)\mathbf{u}_{j} \leq 0, \label{deterministic form of reformulation last time SoC geq}
\end{align}
where $q_{x,\mathrm{final},1} = \sqrt{-\rho^2 \nu_{r}^{\tau}\sum_{k=0}^{\tau-1}\left(b^{k} - a^{k}\right)^{2} \mathrm{Ln}( \tilde{\delta}_{x}^{\tau,\mathrm{final}})}$, $q_{x,\mathrm{final},2} = \sqrt{-\rho^2 \nu_{r}^{\tau}\sum_{k=0}^{\tau-1}\left(b^{k} - a^{k}\right)^{2} \mathrm{Ln}( \delta_{x}^{\tau,\mathrm{final}} -\tilde{\delta}_{x}^{\tau,\mathrm{final}} )}$,
with $0 \leq \tilde{\delta}_{x}^{\tau,\mathrm{final}} \leq 1$, and $0 \leq  \delta_{x}^{\tau,\mathrm{final}} -\tilde{\delta}_{x}^{\tau,\mathrm{final}} \leq 1$. 
\begin{proof}
    The proof is like Proposition \ref{pro. Soc reformulation}.
\end{proof}
\end{proposition}



\begin{proposition}
The feasible domain of the chance constraint \eqref{retailer constraint} can be under-approximated with the deterministic
inequalities
\begin{align}
& -\sum_{j=1}^{N}\mu_{d_{j}} + \mathbf{q}_{g,1} + \sum_{j=1}^{N}\mathbf{u}_{j}  \leq 0,
\label{retailer role in model concatanation leq}\\
& -\bar{g}\mathbf{1}_{\tau} + \sum_{j=1}^{N}\mu_{d_{j}} + \mathbf{q}_{g,2} - \sum_{j=1}^{N}\mathbf{u}_{j}  \leq 0 ,\label{retailer role in model concatanation geq}
\end{align}
where $\mu_{d_{j}} = \left[\mathbf{E}\left\{d_{j}^{0}\right\},\mathbf{E}\left\{d_{j}^{1}\right\},\dots,\mathbf{E}\left\{d_{j}^{\tau-1}\right\}\right]^{T}$, 
$\mathbf{q}_{g,1} = \left[q_{g,1}^0,q_{g,1}^1,\dots,q_{g,1}^{\tau-1}\right]^T$, $\mathbf{q}_{g,2} = \left[q_{g,2}^0,q_{g,2}^1,\dots,q_{g,2}^{\tau-1}\right]^T$,
with $q_{g,1}^t = \sqrt{- \nu_{d}^{t}\sum_{j=1}^{N}\left(f_{j}^{t} - c_{j}^{t}\right)^{2} \mathrm{Ln}(\tilde{\delta}_{g}^{t})}$, $q_{g,2}^t = \sqrt{- \nu_{d}^{t}\sum_{j=1}^{N}\left(f_{j}^{t} - c_{j}^{t}\right)^{2} \mathrm{Ln}( \delta_{g}^{t} - \tilde{\delta}_{g}^{t} )}$, 
$0 \leq \tilde{\delta}_{g}^{t} \leq 1$, and $0 \leq \delta_{g}^{t} - \tilde{\delta}_{g}^{t} \leq 1$. Moreover, $\nu_{d}^{t} = \chi(\hat{G}_{d}^{t})/2$ with $\hat{G}_{d}^{t}$ being the undirected dependency graph of the random variables $d_{1}^{t},\ldots, d_{N}^{t}$.
\end{proposition}

\begin{proof}
We know, that $g_{i}^{t} = d_{i}^{t} - u_{i}^{t}$ and $g^{t} = \sum_{j=1}^{N}g_{j}^{t} = \sum_{j=1}^{N} d_{j}^{t} - \sum_{j=1}^{N} u_{j}^{t}$. We also considered $\mathbf{Pr}\left\{ 0 \leq g^t \leq \bar{g}  \right\} \geq 1-\delta_{g}^{t}$ as retailer constraint. Now like what we have done in Proposition \ref{pro. Soc reformulation}, we have:
\begin{align*}
    &\mathbf{Pr}\left\{ 0 \leq g^t \leq \bar{g}  \right\} \geq 1-\delta_{g}^{t} \Longleftrightarrow\\&  
    \mathbf{Pr}\left\{ \left(0 \leq g^t \right) \bigcap \left( g^t \leq \bar{g} \right) \right\} \geq 1-\delta_{g}^{t} \Longleftrightarrow\\
    &1 - \mathbf{Pr}\left\{ \left(g^t \leq 0 \right) \bigcup \left( g^t \geq \bar{g} \right) \right\} \geq 1-\delta_{g}^{t}\Longleftrightarrow\\&
    \mathbf{Pr}\left\{ \left(g^t \leq 0 \right) \bigcup \left( g^t \geq \bar{g} \right) \right\} \leq \delta_{g}^{t} \Longleftrightarrow\\
    &\mathbf{Pr}\left\{ g^t \leq 0  \right\} + \mathbf{Pr}\left\{  g^t \geq \bar{g}  \right\} \leq \delta_{g}^{t}\Longleftarrow\\
    &\begin{cases}
    	\mathbf{Pr}\left\{ g^t \leq 0  \right\} \leq \tilde{\delta}_{g}^{t}\\
    	\mathbf{Pr}\left\{  g^t \geq \bar{g}  \right\} \leq \delta_{g}^{t} - \tilde{\delta}_{g}^{t}
    \end{cases}
\end{align*}
We know:
\begin{align*}
&\mathbf{Pr}\left\{ g^t \leq 0  \right\} = \mathbf{Pr}\left\{ \sum_{j=1}^{N} d_{j}^{t} - \sum_{j=1}^{N}u_{j}^{t} \leq 0  \right\} = \\&\mathbf{Pr}\left\{ \sum_{j=1}^{N} d_{j}^{t} - \sum_{j=1}^{N}u_{j}^{t} - \sum_{j=1}^{N} \mu_{d_{j}}^{t}\leq - \sum_{j=1}^{N} \mu_{d_{j}}^{t}  \right\}=\\
&\mathbf{Pr}\left\{ \sum_{j=1}^{N} d_{j}^{t}  - \sum_{j=1}^{N} \mu_{d_{j}}^{t}\leq \sum_{j=1}^{N}u_{j}^{t} - \sum_{j=1}^{N} \mu_{d_{j}}^{t}  \right\}
\end{align*}
For identifying a sufficient condition that ensures $\mathbf{Pr}\left\{ g^t \leq 0  \right\} \leq \tilde{\delta}_{g}^{t}$ is satisfied, we use \textit{Chernoff-Hoeffding inequality}. For using \textit{Chernoff-Hoeffding inequality}, it is necessary the following condition satisfied:
\begin{align*}
    \sum_{j=1}^{N}u_{j}^{t} - \sum_{j=1}^{N} \mu_{d_{j}}^{t} \leq 0   
\end{align*}
If the aforementioned condition is met, we can, based on the Chernoff-Hoeffding inequality, state the following:
\begin{align*}
    &\mathbf{Pr}\left\{ \sum_{j=1}^{N} d_{j}^{t}  - \sum_{j=1}^{N} \mu_{d_{j}}^{t}\leq -\left( \sum_{j=1}^{N} \mu_{d_{j}}^{t} - \sum_{j=1}^{N}u_{j}^{t}\right)  \right\} \\&\leq \mathrm{exp}\left( \frac{-\left( \sum_{j=1}^{N} \mu_{d_{j}}^{t} - \sum_{j=1}^{N}u_{j}^{t}\right)^2}{\nu_{d}^{t} \sum_{j=1}^{N}\left( f_{j}^{t} - c_{j}^{t} \right)^2}  \right)
\end{align*}
So it is necessary to have:
\begin{align*}
    &\mathrm{exp}\left( \frac{-\left( \sum_{j=1}^{N} \mu_{d_{j}}^{t} - \sum_{j=1}^{N}u_{j}^{t}\right)^2}{\nu_{d}^{t} \sum_{j=1}^{N}\left( f_{j}^{t} - c_{j}^{t} \right)^2}  \right) \leq \tilde{\delta}_{g}^{t} \Longleftrightarrow\\&
    \frac{-\left( \sum_{j=1}^{N} \mu_{d_{j}}^{t} - \sum_{j=1}^{N}u_{j}^{t}\right)^2}{\nu_{d}^{t} \sum_{j=1}^{N}\left( f_{j}^{t} - c_{j}^{t} \right)^2} \leq \mathrm{Ln}\left( \tilde{\delta}_{g}^{t}  \right) \Longleftrightarrow\\
    &-\left( \sum_{j=1}^{N} \mu_{d_{j}}^{t} - \sum_{j=1}^{N}u_{j}^{t}\right)^2 \leq \nu_{d}^{t} \mathrm{Ln}\left( \tilde{\delta}_{g}^{t}  \right) \sum_{j=1}^{N}\left( f_{j}^{t} - c_{j}^{t} \right)^2 \Longleftrightarrow\\&
    \left( \sum_{j=1}^{N} \mu_{d_{j}}^{t} - \sum_{j=1}^{N}u_{j}^{t}\right)^2 \geq -\nu_{d}^{t} \mathrm{Ln}\left( \tilde{\delta}_{g}^{t}  \right) \sum_{j=1}^{N}\left( f_{j}^{t} - c_{j}^{t} \right)^2
\end{align*}
It is obvious that $-\nu_{d}^{t} \mathrm{Ln}\left( \tilde{\delta}_{g}^{t}  \right) \sum_{j=1}^{N}\left( f_{j}^{t} - c_{j}^{t} \right)^2$ is non-negative, so:
\begin{align*}
    &\left( \sum_{j=1}^{N} \mu_{d_{j}}^{t} - \sum_{j=1}^{N}u_{j}^{t}\right) \geq \sqrt{-\nu_{d}^{t} \mathrm{Ln}\left( \tilde{\delta}_{g}^{t}  \right) \sum_{j=1}^{N}\left( f_{j}^{t} - c_{j}^{t} \right)^2}\\
    &\mathrm{Or}\\
    &\left( \sum_{j=1}^{N} \mu_{d_{j}}^{t} - \sum_{j=1}^{N}u_{j}^{t}\right) \leq -\sqrt{-\nu_{d}^{t} \mathrm{Ln}\left( \tilde{\delta}_{g}^{t}  \right) \sum_{j=1}^{N}\left( f_{j}^{t} - c_{j}^{t} \right)^2}
\end{align*}
But $\left( \sum_{j=1}^{N} \mu_{d_{j}}^{t} - \sum_{j=1}^{N}u_{j}^{t}\right) \geq \sqrt{-\nu_{d}^{t} \mathrm{Ln}\left( \tilde{\delta}_{g}^{t}  \right) \sum_{j=1}^{N}\left( f_{j}^{t} - c_{j}^{t} \right)^2}$ is just valid. So, we consider:
\begin{align*}
   \sum_{j=1}^{N}u_{j}^{t}- \sum_{j=1}^{N} \mu_{d_{j}}^{t} +  \sqrt{-\nu_{d}^{t} \mathrm{Ln}\left( \tilde{\delta}_{g}^{t}  \right) \sum_{j=1}^{N}\left( f_{j}^{t} - c_{j}^{t} \right)^2} \leq 0
\end{align*}
Now we discuss similarly about $\mathbf{Pr}\left\{  g^t \geq \bar{g} \right\} \leq \delta_{g}^{t} - \tilde{\delta}_{g}^{t}$. We know:
\begin{align*}
    &\mathbf{Pr}\left\{  g^t \geq \bar{g} \right\} = \mathbf{Pr}\left\{ \sum_{j=1}^{N}d_{j}^{t} - \sum_{j=1}^{N}u_{j}^{t} \geq \bar{g}  \right\} =\\&  \mathbf{Pr}\left\{ \sum_{j=1}^{N}d_{j}^{t} - \sum_{j=1}^{N}u_{j}^{t} -\sum_{j=1}^{N}\mu_{d_{j}}^{t} \geq \bar{g} -\sum_{j=1}^{N}\mu_{d_{j}}^{t} \right\}
    = \\& \mathbf{Pr}\left\{ \sum_{j=1}^{N}d_{j}^{t} - \sum_{j=1}^{N}\mu_{d_{j}}^{t} \geq \bar{g} -\sum_{j=1}^{N}\mu_{d_{j}}^{t} +\sum_{j=1}^{N}u_{j}^{t} \right\}
\end{align*}
For identifying a sufficient condition that ensures $\mathbf{Pr}\left\{  g^t \geq \bar{g} \right\} \leq \delta_{g}^{t} - \tilde{\delta}_{g}^{t}$ is satisfied, we use \textit{Chernoff-Hoeffding inequality}. For using \textit{Chernoff-Hoeffding inequality}, it is necessary the following condition satisfied:
\begin{align*}
    \bar{g} -\sum_{j=1}^{N}\mu_{d_{j}}^{t} +\sum_{j=1}^{N}u_{j}^{t} \geq 0
\end{align*}
If the aforementioned condition is met, we can, based on the Chernoff-Hoeffding inequality, state the following:
\begin{align*}
    &\mathbf{Pr}\left\{ \sum_{j=1}^{N}d_{j}^{t} - \sum_{j=1}^{N}\mu_{d_{j}}^{t} \geq \bar{g} -\sum_{j=1}^{N}\mu_{d_{j}}^{t} +\sum_{j=1}^{N}u_{j}^{t} \right\}\\&
    \leq \mathrm{exp}\left(  \frac{-\left( \bar{g} -\sum_{j=1}^{N}\mu_{d_{j}}^{t} +\sum_{j=1}^{N}u_{j}^{t}
 \right)^2}{\nu_{d}^{t}\sum_{j=1}^{N}\left( 
 f_{j}^{t} - c_{j}^{t}\right)^2}  \right)
\end{align*}
So it is necessary to have:
{\small
\begin{align*}
    &\mathrm{exp}\left(  \frac{-\left( \bar{g} -\sum_{j=1}^{N}\mu_{d_{j}}^{t} +\sum_{j=1}^{N}u_{j}^{t}
     \right)^2}{\nu_{d}^{t}\sum_{j=1}^{N}\left( 
     f_{j}^{t} - c_{j}^{t}\right)^2}  \right) \leq \delta_{g}^{t} - \tilde{\delta}_{g}^{t} \Longleftrightarrow \\
     & \frac{-\left( \bar{g} -\sum_{j=1}^{N}\mu_{d_{j}}^{t} +\sum_{j=1}^{N}u_{j}^{t}
     \right)^2}{\nu_{d}^{t}\sum_{j=1}^{N}\left( 
     f_{j}^{t} - c_{j}^{t}\right)^2}  \leq \mathrm{Ln} \left\{ \delta_{g}^{t} - \tilde{\delta}_{g}^{t}  \right\} \Longleftrightarrow \\
     & -\left( \bar{g} -\sum_{j=1}^{N}\mu_{d_{j}}^{t} +\sum_{j=1}^{N}u_{j}^{t}\right)^2 \leq \nu_{d}^{t} \mathrm{Ln} \left\{ \delta_{g}^{t} - \tilde{\delta}_{g}^{t}  \right\} \sum_{j=1}^{N}\left( 
     f_{j}^{t} - c_{j}^{t}\right)^2 \\&\Longleftrightarrow \\
     & \left( \bar{g} -\sum_{j=1}^{N}\mu_{d_{j}}^{t} +\sum_{j=1}^{N}u_{j}^{t}\right)^2 \geq -\nu_{d}^{t} \mathrm{Ln} \left\{ \delta_{g}^{t} - \tilde{\delta}_{g}^{t}  \right\} \sum_{j=1}^{N}\left( 
     f_{j}^{t} - c_{j}^{t}\right)^2     
\end{align*}}
It is obvious that $ -\nu_{d}^{t} \mathrm{Ln} \left\{ \delta_{g}^{t} - \tilde{\delta}_{g}^{t}  \right\} \sum_{j=1}^{N}\left( f_{j}^{t} - c_{j}^{t}\right)^2$ is non-negative, so:
\begin{align*}
   &\bar{g} -\sum_{j=1}^{N}\mu_{d_{j}}^{t} +\sum_{j=1}^{N}u_{j}^{t} \geq \sqrt{-\nu_{d}^{t} \mathrm{Ln} \left\{ \delta_{g}^{t} - \tilde{\delta}_{g}^{t} \right\} \sum_{j=1}^{N}\left( 
     f_{j}^{t} - c_{j}^{t}\right)^2 }\\
    &\mathrm{Or}\\
    &\bar{g} -\sum_{j=1}^{N}\mu_{d_{j}}^{t} +\sum_{j=1}^{N}u_{j}^{t} \leq -\sqrt{-\nu_{d}^{t} \mathrm{Ln} \left\{ \delta_{g}^{t} - \tilde{\delta}_{g}^{t}  \right\} \sum_{j=1}^{N}\left( 
     f_{j}^{t} - c_{j}^{t}\right)^2 }
\end{align*}
But $\bar{g} -\sum_{j=1}^{N}\mu_{d_{j}}^{t} +\sum_{j=1}^{N}u_{j}^{t} \geq \sqrt{-\nu_{d}^{t} \mathrm{Ln} \left\{ \delta_{g}^{t} - \tilde{\delta}_{g}^{t}  \right\} \sum_{j=1}^{N}\left( f_{j}^{t} - c_{j}^{t}\right)^2 }$ is just valid. So, we consider:
{\footnotesize
\begin{align*}
    -\bar{g} +\sum_{j=1}^{N}\mu_{d_{j}}^{t} -\sum_{j=1}^{N}u_{j}^{t} + \sqrt{-\nu_{d}^{t} \mathrm{Ln} \left\{ \delta_{g}^{t} - \tilde{\delta}_{g}^{t}  \right\} \sum_{j=1}^{N}\left( f_{j}^{t} - c_{j}^{t}\right)^2 } \leq 0
\end{align*}}
\end{proof}
All the inequalities in \eqref{Soc constraint Concatenate leq}--\eqref{Soc constraint Concatenate geq} and \eqref{deterministic form of reformulation last time SoC leq}
--\eqref{retailer role in model concatanation geq} can be written as
\begin{equation}\label{coupling constraint}
    \sum_{j = 1}^{N} A \mathbf{u}_{j} \leq b,
\end{equation}
with appropriately defined matrices $A$ and $b$, like,
{\normalsize
\begin{align*}
A = 
\begin{bmatrix}
\rho\mathbf{M}_{\tau}\\
\dots \dots\\
-\rho\mathbf{M}_{\tau} \\
\dots \dots\\
\rho\mathbf{1}_{\tau}^{T}\\
\dots \dots\\
- \rho\mathbf{1}_{\tau}^{T}\\
\dots \dots\\
\mathbf{I}_{\tau} \\
\dots \dots\\
- \mathbf{I}_{\tau} \\
\end{bmatrix}
,
b = - 
\begin{bmatrix}
\left(\underline{x}-x^{0}\right)\mathbf{1}_{\tau} - \rho \mathbf{M}_{\tau} \mu_{r} + \mathbf{q}_{x,1}\\
\dots \dots\\
\left(x^{0} - \bar{x}\right)\mathbf{1}_{\tau} + \rho \mathbf{M}_{\tau} \mu_{r}+ \mathbf{q}_{x,2} \\
\dots\dots\\
r-\epsilon-x^{0} - \rho \mathbf{1}_{\tau}^{T} \mu_{r} +q_{x,\mathrm{final},1}
\\
\dots\dots\\
x^{0} - r - \epsilon + \rho \mathbf{1}_{\tau}^{T} \mu_{r} + +q_{x,\mathrm{final},2}
\\
\dots\dots\\
 -\sum_{j=1}^{N}\mu_{d_{j}} + \mathbf{q}_{g,1}\\
\dots\dots\\
 -\bar{g}\mathbf{1}_{\tau} + \sum_{j=1}^{N}\mu_{d_{j}} + \mathbf{q}_{g,2}
\end{bmatrix}\!,
\end{align*}
}
and with explicit local constraint or local decision set
\begin{equation}\label{local constraint}
    \Omega_{i} = \left\{  \mathbf{u}_{i} \in \mathbb{R}^{T}  | \mathbf{0}_{\tau}   \leq \mathbf{u}_{i}  \leq 
\bar{u}_{i}\mathbf{1}_{\tau}  \right\}.
\end{equation}

\noindent\textbf{Reformulation of the cost function.}
We can rewrite the cost function \eqref{cost-pure} as
\begin{equation}
    J_{i} = \mathbf{u}_{i}^{T}  G \mathbf{u}_{i}  
    +T_{i} \mathbf{u}_{i} + \left[ \frac{1}{N}     \sum_{j=1}^{N} \mathbf{u}_{j}^{T}  \right] H \mathbf{u}_{i}  
    + c_{i} \ ,
\end{equation}
where
{\small
\begin{align*}
		G &= 
		\alpha^{dch} \mathbf{I}_{\tau}
		,\hspace{0.1cm} 
            H = 
		Nk_{c}^{N} \mathbf{I}_{\tau}
		,\hspace{0.1cm}
            \\
	    T_{i} &= \left(-K_{ToU} - k_{c}^{N} \mu_{d_{i}} - k_{c}^{N}\sum_{j =1}^{N}\mu_{d_{j}} + \beta^{dch}\mathbf{1}_{\tau}\right)^{T}
	    \\
        c_{i} &= \mu_{d_{i}}^{T}K_{ToU}+\sum_{t=0}^{\tau -1}\Biggl[k_{c}^{N}\mathbf{E}\left\{\sum_{j=1}^{N}d_{i}^{t}d_{j}^{t}\right\}
        \\
        &+\alpha^{dch}\sum_{j=1,j\neq i}^{N}\left(u_{j}^{t}\right)^2
        + \left(-k_{c}^{N}\mu_{d_{i}^{t}}+\beta^{dch} \right)\sum_{j=1,j\neq i}^{N}u_{j}^{t}\Biggr].\\
        K_{ToU} & =  \left[K_{ToU}^{0},K_{ToU}^{1},\dots,K_{ToU}^{\tau-1}\right]^{T},\\
        \mu_{{r}^{t}} &  =  \mathbf{E} \left\{r^{t}\right\},  \mu_{ d_{i}^{t}}  =  \mathbf{E} \left\{d_{i}^{t}\right\}.
	\end{align*}   
}
{\normalsize 
The above under-approximations and reformulation of the cost function gives the following game $\mathcal{G'}$:
}
{\small
     \begin{align*}
      \mathcal{G'} = \begin{cases}
        \textbf{Players:}\  \textrm{A set of residential agents}\  \mathcal{N} = \{1,2,\dots,N\}\\
        \textbf{Strategies of Agents:}\  \mathbf{u}_{i} \hspace{1cm}i\in \mathcal{N}\\
        \textbf{Cost Functions:} \ J_{i}\left(\mathbf{u}_{i},\mathbf{u}_{-i}\right) 
        \hspace{0.75cm} i \in \mathcal{N}\\
        \textbf{Constraints:}
        \begin{cases}
        \textbf{Local:}\ \eqref{local constraint}\\
        \textbf{Coupling\ (deterministic and static):}\ \eqref{coupling constraint} 
        \end{cases} 
        \end{cases}          
    \end{align*}
}
where $\mathbf{u}_{-i} = \mathrm{col}\left(\left\{\mathbf{u}_{j}\right\}_{j \neq i}\right)$. Each agent seeks to minimize its cost function, $J_{i}(\mathbf{u}_{i},\mathbf{u}_{-i})$, while taking into account both the coupling constraint \eqref{coupling constraint} and the local constraint \eqref{local constraint}.
\section{Game Theoretical Analysis}\label{sec.game-analysis}
Utilizing the coupling constraint and local decision set, we define the collective feasible set $\mathcal{U}$ as
\begin{equation}
\mathcal{U}= \Omega \cap \left\{ (\mathbf{y}_{1},\dots,\mathbf{y}_{N}) \in \mathbb{R}^{\tau N} | \sum_{i = 1}^{N}A\mathbf{y}_{i}-b \leq 0\right\},
\end{equation}
where $\Omega = \Omega_{1} \times \Omega_{2} \times \dots \times \Omega_{N}$. Additionally, we define the feasible decision set of agent $i \in \left\{1,\dots,N\right\}$ as
\begin{equation*}
\mathcal{U}_{i}(\mathbf{u}_{-i}) = \left\{ \mathbf{y}_{i} \in \Omega_{i} | A\mathbf{y}_{i} \leq b - \sum_{j=1,j\neq i}^{N} A\mathbf{u}_{-i} \right\}.
\end{equation*}
Our local optimization problem can be formalized as a game theory setup:
\begin{equation}\label{game setup}
    \begin{cases}
    \mathrm{min}_{\mathbf{u}_{i} \in \Omega_{i}} J(\mathbf{u}_{i},\mathbf{u}_{-i})\\
    \mathrm{s.t.} \quad A\mathbf{u}_{i} \leq b - \sum_{j=1,j\neq i}^{N}A\mathbf{u}_{j}
    \end{cases}
    \, \textrm{for all}\  i\in \{1,\dots,N\}
\end{equation}
In this context, we use the notion of Generalized Nash instead of Nash, since the coupling among the agents is not only in their cost functions but also in their constraints.
\begin{definition}[Generalized Nash Equlibrium]
The collective strategy $\mathbf{u}^{*}$ is a generalized Nash equilibrium (GNE) of
the game in \eqref{game setup} if $\mathbf{u}^{*} \in \mathcal{U}$ and for all $i \in \{1,\dots,N\}$
\begin{equation*}
    J_{i}\left(  \mathbf{u}_{i}^{*},\mathbf{u}_{-i}^{*}   \right) \leq  \inf \left\{   J_{i}(\mathbf{y},\mathbf{u}_{-i}^{*})  |  \mathbf{y} \in \mathcal{U}_{i}(\mathbf{u}_{-i}^{*})  \right\}.
\end{equation*}
\end{definition}

\begin{remark}
\label{propetries of 1,2 in main paper}
In our study, we can easily demonstrate that for each $i \in \{1,\dots,N\}$ and $\mathbf{y} \in \mathcal{U}_{-i}$, the local cost function $J_{i}(\cdot, \mathbf{y})$ is strictly convex (since $G$ is positive definite) and continuously differentiable. Additionally, for each agent, its local decision set is non-empty, compact (i.e., closed and bounded), and convex (since it is a hypercube, see \eqref{local constraint}). 
\end{remark}
We also consider the following relax assumption.
\begin{assumption}
\label{slater assumotion}
The collective feasible set satisfies Slater's constraint qualification. 
\end{assumption}

\begin{proposition}
Under Assumption~\ref{slater assumotion}, the GNE exists in our proposed game setup.
\end{proposition}
\begin{proof}
Based on Remark~\ref{propetries of 1,2 in main paper}, Assumption \ref{slater assumotion} and by utilizing Brouwer's fixed point theorem \cite[Proposition 12.7]{palomar2010convex}, it is evident that a GNE exists. 
\end{proof}

Our study aims to identify a subset of the GNE that exhibits favorable properties such as economic fairness and enhanced social stability. Typically, this subset can be determined by solving an appropriate variational inequality, known as a generalized variational inequality (GVI). However, since the local cost functions are continuously differentiable, we utilize a specific type of GVI from \cite{belgioioso2018projected}.
%
 Define the pseduo-gradient $F: \mathcal{U} \rightarrow \mathbb{R}^{\tau N}$ as
\begin{equation*}
F(\mathbf{u})= \mathrm{col}\left(     \left\{     \partial_{\mathbf{u}_{i}}J(\mathbf{u}_{i},\mathbf{u}_{-i})    \right\}_{i \in \{1,\dots,N\} }    \right),  
\end{equation*}
which is single-valued with $F(\mathbf{u}) = \Gamma \mathbf{u} + \Lambda$, where
{\small
\begin{align*}
    \Gamma &= \mathbf{I}_{N \times N} \otimes \left[2\left(G + \frac{H^{T}}{N}\right)\right]
    + \left(\mathbf{1}_{N \times N} - \mathbf{I}_{N \times N}\right) \otimes \left( \frac{H^{T}}{N} \right) \\
    & = \left(2\alpha^{dch}+k_{c}^{N}\right)\left(\mathbf{I}_{N\times N} \otimes \mathbf{I}_{\tau}\right) + k_{c}^{N}\left(\mathbf{1}_{N\times N} \otimes \mathbf{I}_{\tau}\right)\\
    \Lambda &= 
    \begin{bmatrix}
        T_{1}&
        T_{2}&
        \dots&
        T_{N}
    \end{bmatrix}^{T}.
\end{align*}}
    Based on Remark~\ref{propetries of 1,2 in main paper} and \cite[Proposition 12.4]{palomar2010convex}, any solution of standard variational inequality problem $\mathrm{VI}(\mathcal{U}, F)$ is a generalized Nash equilibrium of \eqref{game setup}, However, the converse is not necessarily true.

\begin{proposition} \label{propos:exist and unique variational GNE}
    The variational GNE exists and is unique.
\end{proposition}
\begin{proof}
It has been demonstrated in \cite[Proposition 12.11]{palomar2010convex} that under all the system model properties stated in Remark~\ref{propetries of 1,2 in main paper} and Assumption \ref{slater assumotion}, a sufficient condition for the existence (uniqueness) of a variational GNE in our game setup \eqref{game setup} is $F$ being monotone (strictly monotone). The definition of monotonicity (strict monotonicity) implies that if $\Gamma$ is positive semidefinite (positive definite), then $F$ is monotone (strictly monotone). 
Since $\left(2\alpha^{dch}+k_{c}^{N}\right)\left(\mathbf{I}_{N\times N} \otimes \mathbf{I}_{\tau}\right)$ and $k_{c}^{N}\left(\mathbf{1}_{N\times N} \otimes \mathbf{I}_{\tau}\right)$ are two commuting square matrices, then based on \cite[Theorem 1.3.12]{horn2012matrix}, the  eigenvalues of $\Gamma$ are either $(2\alpha^{dch} + k_{c}^{N} + NK_{c}^{N})$ or $(2\alpha^{dch}+k_{c}^{N})$, both
strictly positive, so $\Gamma$ is positive definite and $F$ is also strictly monotone.
\end{proof}

\begin{remark}\label{remark.strongmonotone}
    $F$ is $\zeta$-strongly monotone with $\zeta \in \left(  0, 2\alpha^{dch} +k_{c}^{N}\right)$. This can be demonstrated by utilizing the definition of strong monotonicity, in a manner analogous to Proposition ~\ref{propos:exist and unique variational GNE}. 
\end{remark}

\begin{remark}
\label{remark.lipschitz}
Based on Lipschitz definition and Gershgorin Theorem, $F$ is $\mathit{l_{F}}$-Lipschitz with $\mathit{l_{F}}  > 2\alpha^{dch} + \left( N+1\right)k_{c}^{N}$.
\end{remark}

As mentioned in \cite{belgioioso2018projected}, projected-gradient algorithms for generalized equilibrium seeking in
aggregative games are preconditioned forward-backward methods. 
Based on this, we propose  a semi-decentralized preconditioned forward-backward algorithm to solve variational inequality problem VI($\mathcal{U},F$),
 as presented in Algorithm~\ref{alg.nashcomputation}. 
In the proposed algorithm, we define $\alpha_{i}$ as real values that belong to the interval $\left(0,\frac{2\zeta}{l_{F}^{2}}\right)$ for all $i \in \{1,\dots,N\}$ and $\gamma \in\left(0,\gamma_{max}\right)$. We set $\gamma_{max} \triangleq \frac{1}{\| \mathbf{1}_{N}^{T} \otimes A \|^{2}}\left[\frac{1}{\alpha_{i,max}}-\frac{1}{2\frac{\zeta}{l_{F}^{2}}}  \right]$.

\begin{algorithm}
\caption{Preconditioned Forward Backward}
\label{alg.nashcomputation}
\SetAlgoLined
\DontPrintSemicolon
\hspace{-0.3cm}\small{\textbf{Initialization:}}\ \footnotesize{ $k \gets 1$, $\mathbf{u}_{i}^{1} \gets \mathbf{u}_{i}^{\mathrm{init}}$, $\lambda^{1} \gets \lambda^{\mathrm{init}}$} \\
\normalsize
\hspace{-0.3cm}\small{\textbf{Repeat}}
\vspace{-0.1cm}
\scriptsize{
\begin{align*}
    &\hspace{0.15cm}\hspace{-0.3cm}{\footnotesize{\textbf{Agents:}}}\hspace{0.1cm}\forall i \in \{1,\dots,N\}\\ &\hspace{0.3cm}\hspace{-0.3cm}\mathbf{u}_{i}^{k+1} \gets \mathrm{proj}_{\Omega_{i}} \left[ \left(  \mathbf{I}_{\tau \times \tau}-2\alpha_{i}\left(G+\frac{H^T}{N}\right)   \right)\mathbf{u}_{i}^{k} - \alpha_{i}A^{T} 
    \lambda^{k} - \alpha_{i}T_{i}^{T}   \right]\\
    &\hspace{0.15cm}\hspace{-0.3cm}{\footnotesize{\textbf{Coordinator:}}}\\
    &\hspace{0.3cm}\hspace{-0.3cm}\lambda^{k+1} \gets \mathrm{proj}_{\mathbb{R}_{\geq 0}^{m}} \left[ 
  \lambda^{k} +\gamma \left(2\ A \sum_{j=1}^{N}\mathbf{u}_{j}^{k+1}-A\sum_{j=1}^{N}\mathbf{u}_{j}^{k}-b   \right) \right]\\
  &\hspace{0.15cm}\hspace{-0.3cm}k  \gets k+1
\end{align*}
}
\normalsize{\hspace{-0.2cm}
\hspace{-0.3cm}\small{\textbf{Until}}}
\scriptsize{{\footnotesize{for all agents}} $\|\mathbf{u}_{i}^{k+1} - \mathbf{u}_{i}^{k}\| \leq \epsilon_{u}^{\mathrm{stop}}$ , $\|\lambda^{k+1} - \lambda^{k}\| \leq \epsilon_{\lambda}^{\mathrm{stop}}$}
\end{algorithm}

Moreover, we denote $\alpha_{i,min}$ and $\alpha_{i,max}$ as the minimum and maximum values of $\alpha_{i}$ across all $i \in \{1,\dots,N\}$ respectively, i.e., $\alpha_{i,min} \triangleq \min_{i \in \{1,\dots,N\}}\alpha_{i}$ and $\alpha_{i,max} \triangleq \max_{i \in \{1,\dots,N\}}\alpha_{i}$.
In this context, the parameters $l_{F}$ and $\zeta$ are determined according to Remark~\ref{remark.strongmonotone} and Remark~\ref{remark.lipschitz}, respectively.

The sequence $\left(\mathrm{col}(\mathbf{u}^{k}, \lambda^{k})\right)_{k=0}^{\infty}$defined by Algorithm 1, with step sizes $\alpha_{i} \in \left(0,\frac{2\zeta}{l_{F}^{2}}\right)$, for all $i \in \{1,\dots,N\}$, and $\gamma \in \left(0,\gamma_{max}\right)$, with $\gamma_{max} \triangleq \frac{1}{\| \mathbf{1}_{N}^{T} \otimes A \|^{2}}\left[  \frac{1}{\alpha_{i,max}}-\frac{1}{2\frac{\zeta}{l_{F}^{2}}}  \right]$
, globally converges to variational GNE based on \cite[Theorem 1]{belgioioso2018projected}. More details about the Algorithm are expressed in Appendix. 
\section{Simulation}\label{sec.simulation}
In this section, we consider a microgrid system consisting of $N = 20$ identical residential users. 
%
The parameters of the cost functions in \eqref{cost-pure} are assigned as $\alpha^{dch} = 8$, $\beta^{dch} = 10$, and $k_c^{N} = 0.015$. Additionally, the conventional time-of-use pricing tariff is provided in Table~\ref{table.tariff}. The initial SoC is $x_0 = 0.5$ with minimum and maximum SoC of $\underline{x} = 0.1$ and $\bar{x} = 0.9$. The battery has a total energy capacity of $E = 20000$ units with efficiency $\eta = 1/E$. The upper bound on the discharging power for each user $i$ is $\bar{u}_{i} = E/\left(N\Delta t\right)$.
In the chance constraints \eqref{Soc constraint}, \eqref{reformulation last time SoC}, and \eqref{retailer constraint}, we set $\delta_{x}^{t} = 0.8$, $\delta_{x}^{\tau,\mathrm{final}} =  0.9$, $r = 0.6$, $\epsilon = 0.05$, $\bar{g} = 600$ and $\delta_{g}^{t} = 0.8$.
The time horizon is $\tau = 24$ hours with time step $\Delta t = 1$ hour.
We consider $d_{i}^{t}$ and $r^t$ to have a bounded support with 25\% deviation around their mean value.
In our configuration, we set $\tilde{\delta}_{x}^{t} = \tilde{\delta}_{x}^{\tau,\mathrm{final}} = \tilde{\delta}_{g}^{t} = 0.05$ and $\nu_{r}^{t} = \nu_{d}^{t} = 1$.
For Algorithm~\ref{alg.nashcomputation}, $l_F = 16.31$, $\gamma = 5.33 \times 10^{-4}$, and $\zeta =4.24$. Additionally, $\alpha =10^{-4} \mathbf{diag}\{p\} \otimes( \mathbf{I}_{\footnotesize{\tau}\times\footnotesize{\tau}})$,
where
$p= [251,311,317,172,131,70,28,152,129,74,32,121,80,223,\\35,71,292,33,251,183]$.

\begin{table}
	\caption{conventional time-of-use pricing tariff.}
	\begin{tabular}{|l|c|c|c|c|c|}
		\hline
		\textbf{Time(t)} & $0-4$ & $5-14$ & $15-16$ & $17-21$ & $22-24$ \\
		\hline
		\textbf{Tariff}($K_{ToU}^{t}$) & 29.45 & 30 & 29.5 & 30.5 & 29.5\\
		\hline
	\end{tabular}
	\label{table.tariff}
\end{table}

\begin{figure}
    \centering
	\includegraphics[width=1\columnwidth]{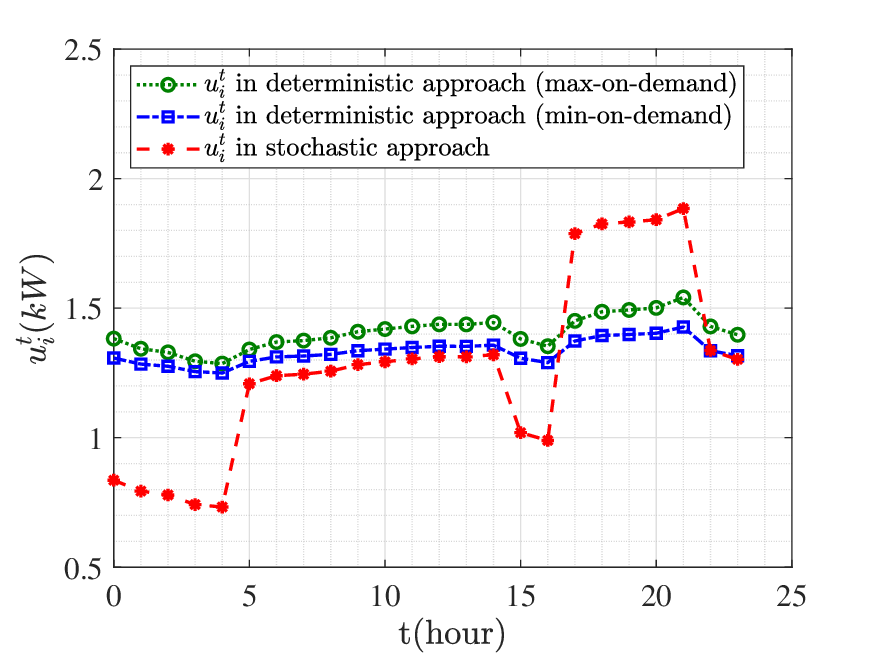}
    \caption{Profile of $u_{i}^{t}$ in stochastic and deterministic approaches.}
    \label{fig:discharging of shared battery in stochastic and deterministic approaches}
\end{figure}

\begin{figure}
    \centering
\includegraphics[width=1\columnwidth]{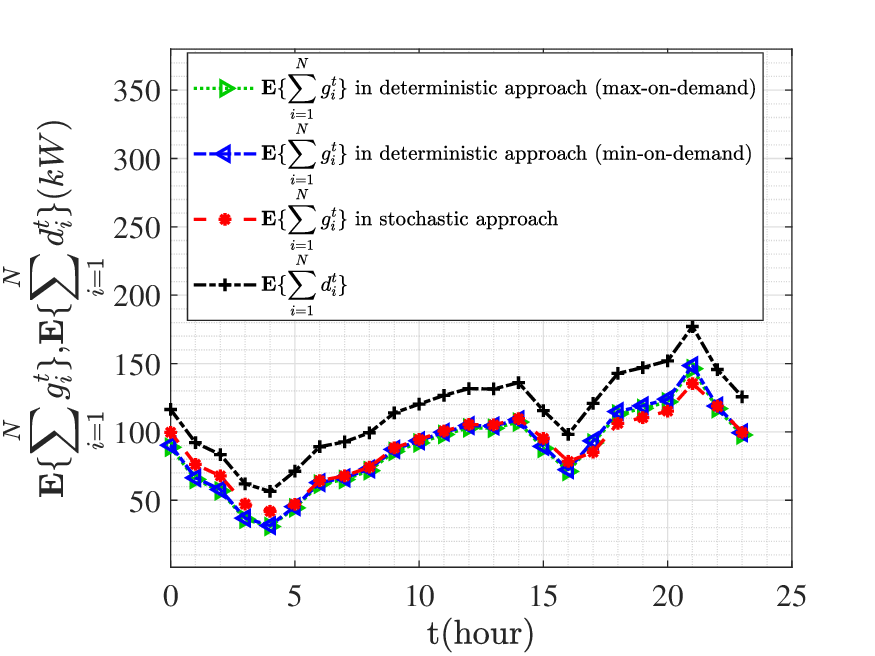}
    \caption{Profile of power exchange of all the agents with the grid in stochastic and deterministic approaches}
    \label{fig:Profile of power exchange of the agent with the grid}
\end{figure} 
According to Proposition~\ref{propos:exist and unique variational GNE}, the variational generalized Nash equilibrium is unique in our stochastic approach. We apply our stochastic method for DSM and compare its results with two deterministic methods. In these deterministic methods, we consider two worst-case scenarios for agent demand (lower bound and upper bound of demands) and use the mean value of renewable energy at each time as its deterministic value.

The profile of $u_{i}^{t}$ in the stochastic and deterministic approaches can be viewed in Figure~\ref{fig:discharging of shared battery in stochastic and deterministic approaches}. This figure illustrates that in periods when prices and demand are relatively low, the stochastic approach utilizes less battery energy compared to the deterministic methods. Conversely, during intervals with higher demand and electricity prices, the battery is employed to a greater extent under the stochastic method. 
Figure \ref{fig:Profile of power exchange of the agent with the grid} illustrates the means of power exchange profile of the agents with the grid.
The results in Figure \ref{fig:Profile of power exchange of the agent with the grid} indicate that the stochastic approach of this paper performs peak shaving more effectively than the deterministic methods. Simulation results also reveal that, while the stochastic approach achieves more effective peak shaving than the deterministic methods, the power exchange profile of the agents with the grid increases more frequently compared to the deterministic approaches. We also compare, numerically, the cost function of our stochastic approach with the deterministic methods for 1,000 random demand values and display the histogram in Figure \ref{fig:histogram comparing the random costs of stochastic and deterministic approaches}. As evident from the figure, the expected cost for our stochastic approach is lower, showing that it incurs lower costs for agents.

\begin{figure}
    \centering
	\includegraphics[width=1\columnwidth]{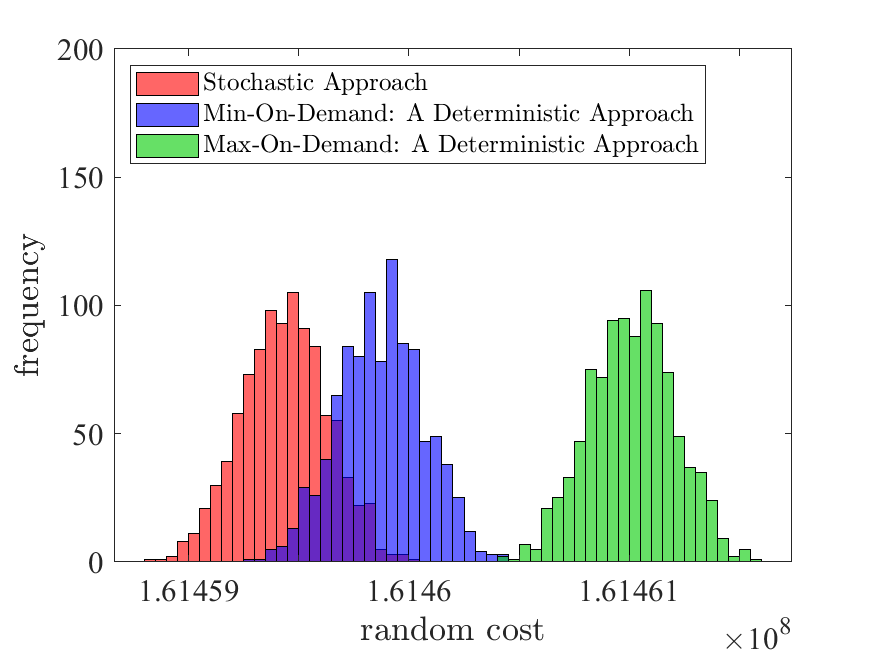}
    \caption{Histogram comparing the random costs of stochastic and deterministic approaches.
    }
    \label{fig:histogram comparing the random costs of stochastic and deterministic approaches}
    \vspace{-0.6cm}
\end{figure}

\section{CONCLUSIONS}
\label{sec.conclusion}
This paper studied a microgrid system where residential users use a shared battery charged through renewable sources and the grid. The model considered uncertain variables including user demand and renewable energy, as well as dynamic and stochastic coupling constraints. The paper used game theory to analyze the Nash equilibrium for demand-side management, examining existence and uniqueness. A semi-decentralized algorithm for Nash seeking was also proposed.
The simulation results demonstrated the advantages of the proposed setup.

\bibliographystyle{ieeetr}
\bibliography{refNew}



\section*{APPENDIX}
\section*{Some Explanation about Proof of Game}
In the section dedicated to game analysis, it is necessary to demonstrate that function $F$ possesses certain characteristics (monotone, strictly monotone, $\eta$-strongly monotone and $l_{f}-\mathrm{Lipschitz}$). Subsequent sections will establish the equivalency of these properties with other concepts in linear algebra (Positive definitive or Positive semi-definite), and the corresponding results.
\subsection{The Monotonicity of $F$}
We start with definition of monotonticity.
\begin{align*}
    &\left(  F(u) - F(v)  \right)^{T}(u-v) \geq 0 \Longleftrightarrow
    \left(  \Gamma u - \Gamma v  \right)^{T}(u-v) \geq 0 \\&\Longleftrightarrow
    \left(  \Gamma \left(u -  v \right) \right)^{T}(u-v) \geq 0 \Longleftrightarrow\\&
    \left(u -  v \right)^T \left( \Gamma^{T} \right)\left(u -  v \right) \geq 0 
    \Longleftrightarrow \left(u -  v \right)^T \left( \Gamma  \right)\left(u -  v \right) \geq 0
\end{align*}
We have ascertained that for $F$ to be monotone, the necessary and sufficient condition is that $F$ must be positive semi-definite. Further, we are also endeavoring to pinpoint conditions leading to $\Gamma$ being positive semi-definite. 
\subsection{strictly Monotonicity of $F$}
We start with definition of monotonticity.
\begin{align*}
    &\left(  F(u) - F(v)  \right)^{T}(u-v) > 0 \Longleftrightarrow
    \left(  \Gamma u - \Gamma v  \right)^{T}(u-v) > 0 \\&\Longleftrightarrow
    \left(  \Gamma \left(u -  v \right) \right)^{T}(u-v) > 0 \Longleftrightarrow\\&
    \left(u -  v \right)^T \left( \Gamma^{T} \right)\left(u -  v \right) > 0 \Longleftrightarrow \left(u -  v \right)^T \left( \Gamma  \right)\left(u -  v \right) > 0
\end{align*}
We have ascertained that for $F$ to be strictly monotone, the necessary and sufficient condition is that $F$ must be positive definite. Further, we are also endeavoring to pinpoint conditions leading to $\Gamma$ being positive definite.
\subsection{$\zeta$-strongly monotonicity of $F$}
We start with definition of $\zeta$-strongly monotone of $F$.
\begin{align*}
    &\left(  F(u) - F(v)  \right)^{T}(u-v) \geq \zeta \| u-v \|^2 \Longleftrightarrow \\&
    \left(  \Gamma u - \Gamma v  \right)^{T}(u-v) \geq \zeta ( u-v )^{T}(u-v) \Longleftrightarrow\\
    &\left(  \Gamma \left(u -  v \right) \right)^{T}(u-v) \geq \zeta ( u-v )^{T}(u-v) \Longleftrightarrow\\&
    \left(u -  v \right)^T \left( \Gamma^{T} -\zeta I  \right)\left(u -  v \right) \geq 0 
    \Longleftrightarrow \\&\left(u -  v \right)^T \left( \Gamma -\zeta I  \right)\left(u -  v \right) \geq 0
\end{align*}
Now we try to find $\zeta$ such that $\Gamma -\zeta I$ becomes positive semi-definite.
\subsection{$l_{f}$-Lipschitz Continuity of $F$}
We start with definition of $l_{f}$-Lipschitz of $F$.
\begin{align*}
    &\|F(u) - F(v)\| \leq l_{f}\|u-v\| \Longleftrightarrow\\&
    \|\Gamma u - \Gamma v\| \leq l_{f}\|u-v\| \Longleftrightarrow
    \|\Gamma u - \Gamma v\|^2 \leq l_{f}^2\|u-v\|^2 \Longleftrightarrow\\
    &\|\Gamma (u - v)\|^2 \leq l_{f}^2\|u-v\|^2 \Longleftrightarrow \\&
    \left( u - v \right)^{T} \Gamma^{T}\Gamma \left( u - v \right) \leq l_{f}^2\left( u - v \right)^{T}\left( u - v \right) \Longleftrightarrow   \\
    &\left( u - v \right)^{T} \left(l_{f}^2I - \Gamma^{2}\right) \left( u - v \right) \geq 0
\end{align*}
Now we try to find $l_{f}$ such that $l_{f}^2I - \Gamma^{2}$ becomes positive semi-definite.
First, we find $\Gamma^2$:
{\small
\begin{align*}
    \Gamma^2 = & \left\{ I_{N \times N} \otimes \left[ 2\left( G + \frac{H^T}{N}   \right) \right] +  \left( \mathbf{1}_{N \times N} - I_{N \times N}  \right) \otimes \frac{H^T}{N}\right\}^2 \\
    = & I_{N \times N} \otimes \left[ 4 \left( \left( \alpha^{dch}+k_{c}^{N} \right)^2 I_{\tau}\right) \right] \\+& \left( \mathbf{1}_{N \times N} - I_{N \times N} \right) \otimes \left[ 2  \left( \alpha^{dch}+k_{c}^{N} \right)k_{c}^{N} I_{\tau} \right]\\
     + & \left( \mathbf{1}_{N \times N} - I_{N \times N} \right) \otimes \left[ 2k_{c}^{N}I_{\tau} \left( \left(  k_{c}^{N} +\alpha^{dch}\right) I_{\tau}\right) \right]\\
     +& \left[ \left( N-2 \right)\mathbf{1}_{N \times N} + I_{N \times N}   \right]\otimes \left( \left(k_{c}^{N}\right)^{2} I_{\tau}\right)
\end{align*}}
{\normalsize
The diagonal elements of $\Gamma^2$ are defined by the expression $4\left(\alpha^{dch} + k_{c}^{N}\right)^2 + \left( N- 1 \right)\left(k_{c}^{N}\right)^2$. Moreover, the summation of the absolute values of off-diagonal elements within a row is given by $4k_{c}^{N}(\alpha^{dch} + k_{c}^{N})(N-1) + (N-2)(N-1)\left(k_{c}^{N}\right)^2$. Therefore, according to the Gershgorin Circle Theorem, the following conditions are necessary.}
{\small
\begin{align*}
    &l_{F}^{2} - 4\left(\alpha^{dch}+ k_{c}^{N}\right)^2 - (N-1)\left(k_{c}^{N}\right)^2 >0
\end{align*}
{\normalsize So}
\begin{align*}
\mathit{l_{F}}  > &\Biggl\{ 4 \left( \alpha^{dch} + k_{c}^{N} \right)^{2} \\& +\left(N-1\right)\left[\left(N-1\right)\left(k_{c}^{N}\right)^{2}+4k_{c}^{N}\left(\alpha^{dch}+k_{c}^{N}\right)\right]\Biggr\}^{\frac{1}{2}}\\
& = 2\alpha^{dch}+(N+1)k_{c}^{N}
\end{align*}}
\section*{Details about algorithm}
In this paper we exploited preconditioned forward-backward splitting
method to compute the generalized Nash equilibria of the game \cite{belgioioso2018projected}. 
Since the convergence characterization of the forward-backward splitting method is well established, the advantage of the proposed design is that global convergence follows provided that some mild monotonicity assumptions on the problem data are satisfied.\\
In the classical Nash Equilibrium concept, each player \(i\) chooses their strategy \(u_i\) to minimize their individual cost function \(J_i(u_i, u_{-i})\), where \(u_{-i}\) denotes the strategies of all other players. In the GNE, there's an added complexity: the players' strategy choices are also subject to some shared/coupled constraints. By writing the KKT condition for finding GNE, this setup leads to a monotone inclusion problem of the form:
\[
0 \in \mathcal{P}(w) + \mathcal{Q}(w)
\]
where $w= \begin{bmatrix}
      \mathbf{u}\\
      \lambda\end{bmatrix}$, that \(\lambda\) is Lagrangian multiplier in the KKT conditions. Moreover,
\begin{align*}
  &\mathcal{P} : \begin{bmatrix}
      \mathbf{u}\\
      \lambda
  \end{bmatrix}\longrightarrow 
  \begin{bmatrix}
      F(\mathbf{u})\\
      b
  \end{bmatrix},\\
  &\mathcal{Q} : \begin{bmatrix}
      \mathbf{u}\\
      \lambda
  \end{bmatrix}\longrightarrow
  \begin{bmatrix}
      \mathbf{N}_{\Omega}(\mathbf{u})\\
      \mathbf{N}_{\mathbb{R}^{m}_{\geq 0}}(\mu)
  \end{bmatrix}+
  \begin{bmatrix}
      0 & (\mathbf{1}_{N}^{T} \otimes A)^{T}\\
      -(\mathbf{1}_{N}^{T} \otimes A) & 0      
  \end{bmatrix}\begin{bmatrix}
      \mathbf{u}\\
      \lambda
  \end{bmatrix}.  
\end{align*}
that \(\mathbf{N}_{S} \) denotes normal cone operator for set \(S\). \(\mathcal{P}\) and \(\mathcal{Q}\) are monotone operators contains the game's cost functions and constraints. In Remark 4 and Remark 5, we establish sufficient conditions under which 
\(F\) is \(\zeta-\mathrm{strongly\  monotone}\) and 
\(l_{f}-\mathrm{Lipschitz} \) . If \(F\) satisfied these properties, we can also show, similar to Lemma 1 in \cite{belgioioso2018projected}, that \(\mathcal{Q}\) is maximally monotone and \(\mathcal{P}\) is \((\zeta/l_{f}^{2})\)-cocoercive.\\
This monotone inclusion problem represents the fixed-point problem associated with finding the GNE. Let \(\mathrm{zer}(\mathcal{P}+\mathcal{Q}) = \{w \in \mathbb{R}^{Nn+m} \mid  0 \in (\mathcal{P}+\mathcal{Q})(w)\}\). Similarly, define \(\mathrm{fix}(V_{\phi}\circ  U_{\phi}) = \{w \in \mathbb{R}^{Nn+m} \mid  w = (V_{\phi}\circ  U_{\phi})(w)\}\). Furthermore, let \(V_{\phi} = (\mathrm{Id}-\phi^{-1}\mathcal{P})\) and \(U_{\phi} = (\mathrm{Id}+\phi^{-1}\mathcal{Q})^{-1}\), where \(\mathrm{Id(.)}\) denotes the identity operator and, positive definite matrix \(\phi\) is preconditioning matrix.
In Lemma 2 of \cite{belgioioso2018projected}, it has been shown that:\\
\begin{equation}
\label{Banach-Picard Eq}
w \in \mathrm{zer}(\mathcal{P} + \mathcal{Q}) \Longleftrightarrow w \in  \mathrm{fix} (V_{\phi}\circ  U_{\phi}),    
\end{equation}

The forward backward algorithm is the Banach–Picard iteration \cite{bauschke2017correction} applied to the mappings  \(V_{\phi}\circ  U_{\phi}\) in \(w \in  \mathrm{fix} (V_{\phi}\circ  U_{\phi})\) , i.e.,
\begin{align}
\label{iteration algorithm}
w^{k+1} = (\mathrm{Id}+ \phi^{-1}\mathcal{Q})^{-1}\circ (\mathrm{Id} - \phi^{-1}\mathcal{P})(w^k)  
\end{align}
In numerical analysis, \(U_{\phi}\) represents a forward step with size and direction defined by \(\phi\), while \(V_{\phi}\) represents a backward step. Directly from the iteration in \eqref{Banach-Picard Eq}, we have that
\begin{align}
\label{Banach-Picard Eq reformulated}
    &(\mathrm{Id}-\phi^{-1}\mathcal{P})(w^k) \in (Id + \phi^{-1}\mathcal{Q})(w^{k+1})\nonumber\\&\Longleftrightarrow -\mathcal{P}(w^k) \in \mathcal{Q}(w^{k+1})+\phi^{-1}(w^{k+1} - w^{k})
\end{align}
The choice of the preconditioning matrix \(\phi\) in \eqref{Banach-Picard Eq reformulated} plays
a key role in the algorithm design (it is important in convergence of our GNE seeking algorithm), and is set base on Theorem 1 in \cite{belgioioso2018projected}. So, We consider $\phi$ like what considered in \cite{belgioioso2018projected}, as
\begin{align*}
    \phi = 
    \begin{bmatrix}
        \alpha^{-1} & -(\mathbf{1}_{N}^{T} \otimes A)^{T}\\
        -(\mathbf{1}_{N}^{T} \otimes A) & \gamma^{-1} \mathbf{I}
    \end{bmatrix}
\end{align*}

where $\alpha = \mathrm{diag}\left(\alpha_{1},\alpha_{2},\dots,\alpha_{N}\right)\otimes\mathbf{I}_{\tau}$ and coefficients $\{\alpha_{i}\}_{i=1}^{N}$ and $\gamma$ are chosen such that $\phi$ is positive definite (based on design guidelines and Theorem 1 in \cite{belgioioso2018projected}). By substituting the the above \(\phi\) in \eqref{Banach-Picard Eq reformulated} and doing some manipulation we can easily reach to
\\ 
{\footnotesize
\begin{align*}
    &-\begin{bmatrix}
      F\left(\mathbf{u}^{k}\right)\\
      b
    \end{bmatrix}
    \in
    \begin{bmatrix}
        \mathbf{N}_{\Omega}\left(\mathbf{u}^{k+1}\right)\\
        \mathbf{N}_{\mathbb{R}^{m}\geq 0}\left(\lambda^{k+1}\right)
    \end{bmatrix}
    +\\
    &\begin{bmatrix}
        \left(\mathbf{1}_{N}^{T} \otimes A\right)^{T} \lambda^{k+1}+ \alpha^{-1}\left(\mathbf{u}^{k+1}-\mathbf{u}^{k}\right) - \left(\mathbf{1}_{N}^{T} \otimes A\right)^{T}\left(\lambda^{k+1} - \lambda{k}\right)\\
        -\left(\mathbf{1}_{N}^{T} \otimes A\right) \mathbf{u}^{k+1} - \left(\mathbf{1}_{N}^{T} \otimes A\right)\left(\mathbf{u}^{k+1} - \mathbf{u}^{k}\right) + \gamma^{-1}\left(\lambda^{k+1} - \lambda^{k}\right)
    \end{bmatrix}    
\end{align*}}
First we consider the following equations
\begin{align}
    \label{preLagrangianUpdate}
    -b =& -\left(\mathbf{1}_{N}^{T} \otimes A\right) \mathbf{u}^{k+1} - \left(\mathbf{1}_{N}^{T} \otimes A\right)\left(\mathbf{u}^{k+1} - \mathbf{u}^{k}\right) \nonumber\\&+ \gamma^{-1}\left(\lambda^{k+1} - \lambda^{k}\right)
\end{align}
\begin{align}
    \label{preUupdate}
    - F\left(\mathbf{u}^{k}\right) =& \left(\mathbf{1}_{N}^{T} \otimes A\right)^{T} \lambda^{k+1}+ \alpha^{-1}\left(\mathbf{u}^{k+1}-\mathbf{u}^{k}\right) \nonumber\\&- \left(\mathbf{1}_{N}^{T} \otimes A\right)^{T}\left(\lambda^{k+1} - \lambda^{k}\right)
\end{align}
Based on  \eqref{preLagrangianUpdate}, we can find
\begin{align}
    \label{LagrangianUpdate}
    \lambda^{k+1} = \lambda^{k} + \gamma\left(2\left(\mathbf{1}_{N}^{T} \otimes A\right)\mathbf{u}^{k+1} - \left(\mathbf{1}_{N}^{T} \otimes A\right)\mathbf{u}^{k} -b\right)
\end{align}
By substituting \eqref{LagrangianUpdate} in \eqref{preUupdate}, we can reach to
\begin{align}
    \label{Uupdate}
    \mathbf{u}^{k+1} = \mathbf{u}^{k} - \alpha\left(F\left(\mathbf{u}^{k}\right)+ \left(\mathbf{1}_{N}^{T} \otimes A\right) \lambda^{k}\right)
\end{align}
As it has been mentioned in main manuscript, $F\left(\mathbf{u}^{k}\right) = \Gamma \mathbf{u}^{k} + \Lambda$, so \eqref{Uupdate} can express as
\[\mathbf{u}^{k+1} \longleftarrow  \left(  \mathbf{I}- \alpha \Gamma  \right)\mathbf{u}^{k} - \alpha \left(\mathbf{1}_{N}^{T} \otimes A\right)^{T} 
    \lambda^{k} - \alpha \Lambda .\]
Now by applying normal cone operator to $\mathbf{u}^{k+1}$ we have
\[\mathbf{u}^{k+1} \longleftarrow \mathrm{proj}_{\Omega} \left[ \left(  \mathbf{I}- \alpha \Gamma  \right)\mathbf{u}^{k} - \alpha \left(\mathbf{1}_{N}^{T} \otimes A\right)^{T} 
    \lambda^{k} - \alpha \Lambda   \right].\]
Also by applying normal cone operator to $\lambda^{k+1}$ we have
{\footnotesize
\begin{align}
    \lambda^{k+1} &\gets \mathrm{proj}_{\mathbb{R}_{\geq 0}^{m}} \left[ 
  \lambda^{k} + \gamma \left( 2\left(\mathbf{1}_{N}^{T} \otimes A\right)\mathbf{u}^{k+1}-\left(\mathbf{1}_{N}^{T} \otimes A\right)\mathbf{u}^{k}-b   \right) \right]
\end{align}}
So, Algorithm \ref{alg.nashcomputation_old} can be used as a centralized algorithm for GNE seeking. \\
The centralized algorithm presented was modified to a semi-decentralized one that employs a coordinator to become GNE seeking computationally scalable for large scale networks.
To arrive at the semi-decentralized algorithm, we undertook these steps. Let's delve deeper into the following update term in Algorithm \ref{alg.nashcomputation_old}
\[\mathbf{u}^{k+1} \longleftarrow \mathrm{proj}_{\Omega} \left[ \left(  \mathbf{I}- \alpha \Gamma  \right)\mathbf{u}^{k} - \alpha \left(\mathbf{1}_{N}^{T} \otimes A\right)^{T} 
    \lambda^{k} - \alpha \Lambda   \right].\]
where $\mathbf{u} = [\mathbf{u}_{1}^{T},\mathbf{u}_{2}^{T},\dots,\mathbf{u}_{N}^{T}]^{T}$.\\
By expanding the right hand side, we found that $\mathbf{u}_{i}$ is updated as
{\footnotesize
\[\mathbf{u}_{i}^{k+1} \longleftarrow \mathrm{proj}_{\Omega_{i}} \left[ \left(  \mathbf{I}_{\tau \times \tau}-2\alpha_{i}(G+\frac{H^T}{N})   \right)\mathbf{u}_{i}^{k} - \alpha_{i}A^{T} 
    \lambda^{k} - \alpha_{i}T_{i}^{T}   \right].\]}

Moreover, we can easily find that $\lambda$ is updated as
{\small
\[
\lambda^{k+1} \longleftarrow \mathrm{proj}_{\mathbb{R}_{\geq 0}^{m}} \left[ 
  \lambda^{k} +\gamma \left(2\ A \sum_{j=1}^{N}\mathbf{u}_{j}^{k+1}-A\sum_{j=1}^{N}\mathbf{u}_{j}^{k}-b   \right) \right]
\]}

Fig.1 in the paper shows the information flow between the agents and coordinator. As we can see, the coordinator does not need to know the local objectives of the agents and only receives the data of $\mathbf{u}_{j}^{k}$ and $\mathbf{u}_{j}^{k+1}$ from the agents and then updates the value of $\lambda$ accordingly. Such a semi-decentralized scheme has been also used in some other studies like \cite{belgioioso2017semi,belgioioso2021semi}.
So we propose a semi-decentralized algorithm, as presented in Algorithm \ref{alg.nashcomputation}.\\
\begin{algorithm}[h]
\caption{Preconditioned Forward Backward (Centralized Algorithm)}
\label{alg.nashcomputation_old}
\SetAlgoLined
\DontPrintSemicolon
\small{\textbf{Initialization:}}\ \footnotesize{ $k \gets 1$, $\mathbf{u}_{i}^1 \gets \mathbf{u}^{\mathrm{init}}$, $\lambda^1 \gets \lambda^{\mathrm{init}}$}\;
\normalsize
\small{\textbf{Repeat}}
\vspace{-0.1cm}
\scriptsize{
\begin{align*}
    \mathbf{u}^{k+1} &\gets \mathrm{proj}_{\Omega} \left[ \left(  \mathbf{I}- \alpha \Gamma  \right)\mathbf{u}^{k} - \alpha \left(\mathbf{1}_{N}^{T} \otimes A\right)^{T} 
    \lambda^{k} - \alpha \Lambda   \right]\\
    \lambda^{k+1} &\gets \mathrm{proj}_{\mathbb{R}_{\geq 0}^{m}} \left[ 
  \lambda^{k} + \gamma \left( 2\left(\mathbf{1}_{N}^{T} \otimes A\right)\mathbf{u}^{k+1}-\left(\mathbf{1}_{N}^{T} \otimes A\right)\mathbf{u}^{k}-b   \right) \right]\\
  k & \gets k+1
\end{align*}
}
\normalsize{\hspace{-0.2cm}
\small{\textbf{Until}}}
\scriptsize{$\|\mathbf{u}^{k+1} - \mathbf{u}^{k}\| \leq \epsilon_{u}^{\mathrm{stop}}$,
    $\|\lambda^{k+1} - \lambda^{k}\| \leq \epsilon_{\lambda}^{\mathrm{stop}}$}
\end{algorithm}

%

\end{document}